\documentclass[pra,reprint]{revtex4-2}
\pdfoutput=1 
\usepackage[T1]{fontenc}
\usepackage[utf8]{inputenc}
\usepackage{amsmath}
\usepackage{amssymb}
\usepackage{amsthm}
\usepackage{hyperref}
\usepackage[caption=false]{subfig}
\usepackage{mathtools}
\usepackage{physics}
\usepackage{dsfont}
\usepackage{tikz}
\newtheorem{theorem}{Theorem}

\begin{document}
\title{Quantum error correction with higher Gottesman-Kitaev-Preskill codes:\\minimal measurements and linear optics}
\date{\today}
\author{Frank Schmidt}
\email{fschmi@students.uni-mainz.de}
\affiliation{Institute of Physics, Johannes Gutenberg-Universit\"at Mainz, Staudingerweg 7, 55128 Mainz, Germany}
\author{Peter van Loock}
\email{loock@uni-mainz.de}
\affiliation{Institute of Physics, Johannes Gutenberg-Universit\"at Mainz, Staudingerweg 7, 55128 Mainz, Germany}
\begin{abstract}
 We propose two schemes to obtain Gottesman-Kitaev-Preskill (GKP) error syndromes by means of linear optical operations, homodyne measurements and GKP ancillae.
 This includes showing that for a concatenation of GKP codes with a $[n,k,d]$ stabilizer code only $2n$ measurements are needed in order to obtain the complete syndrome information,  significantly reducing the number of measurements in comparison to the canonical concatenated measurement scheme and at the same time generalizing linear-optics-based syndrome detections to higher GKP codes.
 Furthermore, we analyze the possibility of building the required ancilla states from single-mode states and linear optics.
 We find that for simple GKP codes this is possible, whereas for concatenations with qubit Calderbank-Shor-Steane (CSS) codes of distance $d\geq3$ it is not. 
 We also consider the canonical concatenated syndrome measurements and propose methods for avoiding crosstalk between ancillae. 
 In addition, we make use of the observation that the concatenation of a GKP code with a stabilizer code forms a lattice in order to see the analog information decoding of such codes from a different perspective allowing for semi-analytic calculations of the logical error rates.
\end{abstract}
\maketitle

\section{Introduction}
In the last few years large interest arose in bosonic quantum error correcting schemes, which encode a finite dimensional system within a harmonic oscillator, such as cat and GKP codes \cite{catcode,gkp}.
This growing interest for such codes came from experiments demonstrating first implementations of these codes \cite{cat_experiment,Fluehmann2019,Campagne-Ibarcq2020} and partly already outperforming simple encodings, although the codes were proposed already two decades ago.
As these codes even allow for error correction with a single oscillator mode they are very hardware-efficient.
However, the GKP codes are only able to correct small displacement errors and therefore concatenations with stabilizer codes \cite{kosuke_minimal_measurement,PhysRevA.101.012316,PhysRevA.99.032344,PhysRevX.8.021054,PhysRevA.102.052408} are often considered in order to correct larger shifts. 
The analog syndrome information of individual GKP codes has gained a lot of attention as it helps to further boost the error-correction capability of the code concatenation, because even for a code of distance $d=3$ it allows for correcting some two-qubit errors.

GKP codes are now also considered for quantum communication, since they can be encoded in an electromagnetic light field, which is the ideal long-distance quantum information carrier, and so have been shown to almost achieve the capacity of the loss channel in the low loss regime \cite{gkp_capacity}.
Furthermore, for quantum communication one only needs Clifford gates and Pauli-measurements which can be implemented in the GKP encoding with Gaussian optics and homodyne measurements.
Recently concatenations with qubit stabilizer codes have been considered for communication \cite{kosuke_gkp_repeater} making also use of the analog information in the GKP error syndrome \cite{analog_gkp_repeater}.

In this paper we primarily describe the GKP codes by making use of their stabilizer formulation, because this allows us to simply generalize results from the usual square lattice GKP code to more general lattices and it is also useful for the concatenation with high-level codes, which we assume to be  qubit (qudit) stabilizer codes.
 We show that it is not necessary to first perform the GKP syndrome measurements and later those of the stabilizer code independently as it is usually done in the literature. Instead it is possible to find a joint minimal set of stabilizer generators for the concatenation of both codes which can then be measured  reducing the overhead of necessary ancilla states.
 Related to this result, we propose two explicit methods for obtaining this syndrome information without inline squeezing operations and based on passive linear optics.
 In particular, our linear-optics schemes for the error correction syndrome detections include those of the higher GKP codes, thus extending existing linear-optics schemes for sole GKP qubit syndrome detection.
 We show that the error correcting properties of a code remain invariant under (passive) linear optical transformations for isotropic displacement noise.
 Additionally, we also discuss the possibility of generating the ancilla states necessary for error correction with linear optics and show that it is impossible to generate codewords of such a high-level GKP qubit code with code distance $d\geq3$ by employing rectangular single-mode grid states and linear optics.
 These results are not in contradiction and complementary to the results from Ref. \cite{xanadu_linearoptics} who consider additional GKP states which are then measured, while we do not assume such additional GKP states. We also discuss some other results concerning the possibility of building GKP-type states with passive linear optics, namely for GKP Bell states composed of two general (multi-mode) GKP codes or codewords assuming that two copies of suitable codes or codewords are already experimentally accessible.   
  
Moreover, we also discuss the possibility of performing syndrome measurements of the higher level code following the canonical measurement approach in such a way that there is no error propagation from one ancilla to another one.
Finally, we demonstrate how one can systematically calculate the performance of the concatenation of GKP qudits with a high level code when making use of the analog syndrome information in a semi-analytic way. 

The paper is structured as follows. In Sec. \ref{sec:two} we review qudits and GKP codes, in Sec. \ref{sec:three} we give a brief review about different schemes for obtaining the GKP syndrome information. In Sec. \ref{sec:four} we discuss the minimal number of measurements for higher GKP codes and propose a linear-optical realization based on error correction by teleportation. In Sec. \ref{sec:knill-glancy} we propose another linear-optical realization of the minimal set of measurements and in Sec. \ref{sec:six} we discuss methods for avoiding error propagation between ancillas when performing stabilizer measurements. Finally, we compare the different methods of obtaining the syndrome information in Sec. \ref{sec:seven} and conclude in Sec. \ref{sec:conclusion}.
\section{Background}\label{sec:two}
\subsection{Qudits}
We refer to a quantum system represented by a finite dimensional Hilbert space of dimension $D$ as a qudit of dimension $D$. Furthermore, we label states in the $Z$-basis by elements of $\mathbb{Z}_D$. For these qudits we can generalize the Pauli operators as
\begin{align}
	X_D=\sum_{j=0}^{D-1}\ket{j+1 \mod D}\bra{j}\,,\\Z_D=\sum_{j=0}^{D-1} \exp(i \frac{2\pi}{D}j)\ket{j}\bra{j}\,,\\
	Z_DX_D=\exp(i\frac{2\pi}{D})X_DZ_D\,.\label{eq:pauli_commutation}
\end{align}
For qudits we can then give  $D^2$ basis elements for all operators, taking the form $P^{rs}:=X_D^r Z_D^s$ with $r,s\in \mathbb{Z}_D$. For brevity we drop the index $D$ in Pauli operators.
When neglecting global phase information, it is possible to map Pauli operators acting on $n$ qudits onto
 $\mathbb{Z}_D^{2n}$ via 
\begin{align}
	\phi(X_1^{r_1}Z_1^{s_1}\dots X_n^{r_n}Z_n^{s_n})=\left(r_1,\dots,r_n|s_1,\dots,s_n\right)\,.
\end{align}
Using equation \ref{eq:pauli_commutation} we see that
\begin{align}
P^{rs}P^{r's'}=\exp\left(-i\frac{2\pi}{D}\omega((r,s),(r',s'))\right) P^{s',r'}P^{s,r}\,,
\end{align} 
where $\omega(\cdot,\cdot)$ is the canonical symplectic form given by $\omega((r,s),(r',s')=r\cdot s'-s\cdot r'$. Thus, two Pauli operators commute if and only if the symplectic form of the two symplectic representations of the Pauli operators vanishes modulo $D$.

Stabilizer codes (see \cite{nielsen_chuang_2010,qudit_stabilizer} for more details) are defined by an abelian subgroup $S$ of the Pauli group which acts as the identity within the code space. Given such a group it is possible to find a small set generating the whole group. For the special case of prime $D$ there is the nice relation that the number of stabilizer generators is equal to $n-k$, where $n$ is the number of physical qudits and $k$ is the number of encoded qudits. However, for non-prime $D$ we can have up to $2n$ stabilizer generators \cite{qudit_stabilizer}. The code distance of a stabilizer code is given by the lowest weight element in $\mathcal{N}(S)/S$, where $\mathcal{N}(S)$ denotes the normalizer of $S$.
It is quite convenient to give a stabilizer code by a $l \times 2n$ matrix given by the symplectic representation of the $l$ stabilizer generators. For CSS-codes this matrix can be brought to the following form

\begin{equation}
H=\begin{pmatrix}
H_X & 0 \\ 
0 & H_Z
\end{pmatrix} \,.
\end{equation}

Thus, bit- and phase-flips can be corrected independently.

\subsection{GKP codes}

GKP codes \cite{gkp} encode $n$ qudits within the phase space of a harmonic oscillator with $n$ modes. These codes can be understood as stabilizer codes, where the code space is stabilized by by a discrete, abelian subgroup of the continuous Weyl-Heisenberg group. \footnote{We denote this continuous Weyl-Heisenberg group simply as Weyl-Heisenberg group, while we refer to the discrete Weyl-Heisenberg group as Pauli-group.}

The elements of the continuous Weyl-Heisenberg group for $n$ modes can be given as $U(\theta,\alpha,\beta)=\exp(i \theta)\exp(i\sqrt{2\pi}\sum_{j=1}^{n}(\alpha_j \hat{q}_j+\beta_j \hat{p}_j))$ with real numbers $\alpha,\beta \in \mathbb{R}^n$ and $\hat{q}$ and $\hat{p}$ denote the position and momentum operators fulfilling $[\hat{q},\hat{p}]=i \hbar$; in this article we set $\hbar=1$. Thus this group is isomorphic to $U(1)\times \mathbb{R}^{2n}$. The commutation relation of two group elements is given by
\begin{align}
U(\theta_1,\alpha_1,\beta_1)U(\theta_2,\alpha_2,\beta_2)=U(\theta_2,\alpha_2,\beta_2)U(\theta_1,\alpha_1,\beta_1)\nonumber\\ \times \exp(-i 2\pi \omega((\alpha_1,\beta_1),(\alpha_2,\beta_2)))\,,
\end{align}
where $\omega(\cdot,\cdot)$ is the canonical symplectic product already introduced in the previous qudit section extended to real numbers and can be obtained by the Baker-Campbell-Hausdorff formula. In order to obtain commuting operators, we need to find elements in $\mathbb{R}^{2n}$ whose symplectic product gives pairwise an integer. We will refer to the parameterization via $\mathbb{R}^{2n}$ as phase space or symplectic representation.
In order to encode a finite-dimensional system in the $2n$ dimensional code space, we need $2n$ independent stabilizer generators which we use as a definition for the stabilizer group.
If we have found those elements in $\mathbb{R}^{2n}$, then we know that also all elements in the lattice $\mathcal{L}$ generated by the $2n$ independent vectors in $\mathbb{R}^{2n}$ also correspond to commuting operators due to the linearity of the symplectic product \footnote{Every point in the lattice corresponds to a member of the stabilizer group. However, we defined the stabilizer group by products of stabilizer generators and when converting the product into a single Weyl-Heisenberg operator by using the Baker-Campbell-Hausdorff formula a phase of $\pm1$ can appear. This sign ambiguity, however, does not affect our results. }. The set of operators commuting with all stabilizers corresponds to the dual lattice $\mathcal{L}^\bot$ (with respect to the symplectic form). Thus $\mathcal{L}^\bot/\mathcal{L}$ give logical operators and therefore we can define the code distance (with respect to the euclidean norm), analogously to qudit stabilizer codes, as the minimum weight of non-trivial elements in $\mathcal{L}^\bot/\mathcal{L}$ giving the smallest error commuting with all stabilizers.

As an example let us consider the well known square lattice code. The stabilizer generators are given by
\begin{equation}
	\exp(-i\sqrt{2\pi D}\hat{p})\,,\quad \exp(i \sqrt{2\pi D}\hat{q})\,,
\end{equation}
with logical operators
\begin{equation}
	\overline{X}=\exp(-i\sqrt{\frac{2\pi}{D}}\hat{p})\,, \quad \overline{Z}=\exp(i \sqrt{\frac{2\pi }{D}}\hat{q})\,.
\end{equation}

Thus all displacement errors smaller than $\sqrt{\frac{\pi}{2D}}$ can be corrected. However, notice that the logical states $\ket{j}$ in the $Z$-basis are given as
\begin{equation}
	\ket{j}=\sum_{k\in \mathbb{Z}}\ket{\hat{q}=\sqrt{2\pi D}\left(k+\frac{j}{D}\right)}\,.
\end{equation} 
The codewords consist of a infinite series of delta peaks in position or momentum representation such that the states are unphysical, because they are not normalizable and have infinite energy. Thus one needs to consider approximate GKP states, where we replace the delta peaks by narrow Gaussian peaks and we also consider an overall Gaussian envelope in order to make the state normalizable. 
Such a state can be obtained by applying coherent, Gaussian displacements on an ideal codewords. There are multiple approximations known in the literature which have been shown to be equivalent \cite{gkp_equivalence}.
 In this article, we replace the coherent Gaussian displacements by incoherent ones, simplifying the calculations.
This can be understood as the result of a unphysical limit of a twirling operation \cite{wang2019quantum,PhysRevA.103.022404} acting on a state with coherent displacements similar to the qubit case where it is also possible to reduce arbitrary noise to Pauli channels by applying twirling operations. Thus the resulting state is noisier such that we obtain a conservative estimate of the error correction properties.

One main advantage of this GKP encoding is that all Clifford operations acting on the GKP code can be implemented by Gaussian operations. Additionally, Pauli-measurements can be implemented by using homodyne-measurements. Furthermore, GKP syndrome measurements, which can be implemented by GKP states and Gaussian operations, applied to the vacuum state are known to produce states that can be distilled to magic states \cite{all_gaussian}. Thus, the generation of the GKP states is the only needed non-Gaussian element for a universal set of quantum gates. An all-Gaussian system can be simulated efficiently \cite{CVGottesmann-Knill}.

\section{Review of syndrome measurements}\label{sec:three}
We consider a concatenation of GKP qudits with qudit stabilizer codes. We refer to the syndrome measurement where we obtain information about the small shifts needed for correcting the GKP qudit as GKP syndrome, while we will refer to the syndrome obtained by measuring the stabilizer generators as stabilizer syndrome.
\subsection{Stabilizer syndrome}
The syndrome of a stabilizer code which encodes $k$ logical qudits into $n$ physical qudits is formally obtained by measuring all $n-k$ stabilizer generators (for $D$ prime, otherwise up to $2n$).
When coupling ancilla qubits with data qubits for obtaining the code syndrome it is highly desirable that every ancilla qubit only couples with a single data qubit in order to prevent a single error of the ancilla propagating onto multiple data qubits. One such scheme is the Steane error correction \cite{Steane_EC} where the $n$ data qubits are coupled with $2n$ ancilla qubits by transversal CNOT gates. The CNOTs act as the identity on the logical level for this ancilla such that we learn the error syndrome but gain no information about the encoded quantum information. In the special case of a CSS code the $2n$ qudit ancilla states can be decomposed into the two logical codewords $\ket{\overline{+}}$ ($\ket{\overline{0}}$) being target (control) of the transversal CNOTs and measured in the $Z$ ($X$) basis.

A different approach is the so called Knill scheme\cite{Knill_teleportation}, where the ancilla is given by a logical Bell state and Bell measurements are applied to the data qubits and one half of the logical Bell state. In the original paper the scheme was only proven for qubits, but in App. \ref{app:knill} we show that it also works for $D>2$ CSS-codes.
\subsection{GKP syndrome}
Let us begin to review the different known methods for obtaining the GKP syndrome. The schemes can be put into two categories. On the one hand we have sequential measurements as the one proposed in the original GKP paper \cite{gkp}, which is inspired by the Steane error correction scheme for CSS qubit codes \cite{Steane_EC}, and further improved schemes reducing the experimental resources \cite{Knill_Glancy,gkp_offline}. On the other hand we have a teleportation based scheme \cite{towards_grid_states,all_gaussian} inspired by Knill's error correction by teleportation \cite{Knill_teleportation} which only started to gain more interest recently \cite{bosonic_gate_teleportation}.  

\subsubsection{Steane scheme}
Now let us further discuss the sequential scheme. For square GKP qubits the Steane error correction scheme (Fig. \ref{fig:steane_ec}) was proposed for performing the syndrome measurement. First we have one code block containing the data and two ancilla code blocks being in the $\ket{+}$ and $\ket{0}$ state.
In order to obtain the syndrome information of the modular position a  CNOT is applied to the data code (control) and the first ancilla code (target) and the mode of the first ancilla code is measured in the position quadrature. Similarly, we obtain the modular momentum stabilizers by applying a CNOT to the second ancilla code (control) and data code (target) and the mode of the ancilla code is measured in the momentum quadrature. In the codespace this acts as the identity and therefore by obtaining the error syndrome we do not obtain information about the logical state.
For the square GKP code CNOT gates are implemented by CSUM gates ($\exp\left(-i \hat{q}_1\hat{p}_2\right)$) which can be decomposed into two beam splitters and two squeezing operations. From an experimental point of view arbitrary passive linear optical transformations, decomposable into beam splitters and phase shifters, are easy to implement while squeezing operations are not that simple to implement (highest squeezed vacuum state 15 dB \cite{highest_exp_squeezing}). Furthermore, it is hard to implement an operation which acts as the squeezing operation on arbitrary input states. Thus these squeezing operations are typically implemented via gate teleportation with an, ideally infinitely, squeezed ancilla state \cite{squeezing_teleportation} and have already been used for implementing a CSUM gate experimentally \cite{loock_experimental_csum}. However, infinitely squeezed vacuum states are unphysical and can only be approximated by highly squeezed vacuum states resulting in an approximation error. Thus, it is beneficial to avoid inline squeezing and use offline squeezing whenever possible. 

\subsubsection{Knill-Glancy scheme}
The Knill-Glancy scheme \cite{Knill_Glancy} (Fig. \ref{fig:knill_glancy}) was proposed for a square lattice (although it is easy to see that it also works for rectangular lattices) GKP code and it can be understood as a variation of Steane error correction where the CSUM gate is replaced by a 50:50 beam splitter followed by a squeezing operation with a squeezing factor $\sqrt{2}$.
Independently from our work, it was recently shown in Ref \cite{gkp_offline} that the Knill-Glancy scheme is equivalent to a scheme where no inline squeezing is used (Fig. \ref{fig:improved_knill_glancy}), but one of the two ancilla GKP states is squeezed by a factor of $\sqrt{2}$. In section \ref{sec:knill-glancy} we will show that these improvements also work for arbitrary GKP codes. Furthermore, this improves the noise introduced by finite squeezing and there also exists a similar scheme which also gives the syndrome information of a high level CSS code.
\begin{figure}
	\subfloat[]{\label{fig:steane_ec}\includegraphics{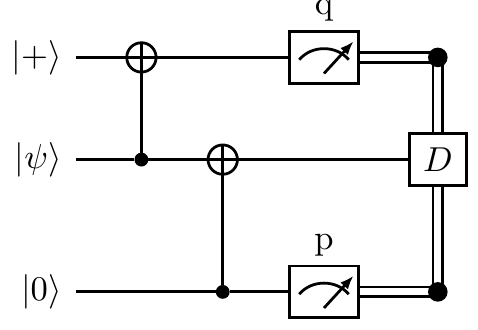}}\\
	\subfloat[]{\label{fig:knill_glancy}\includegraphics{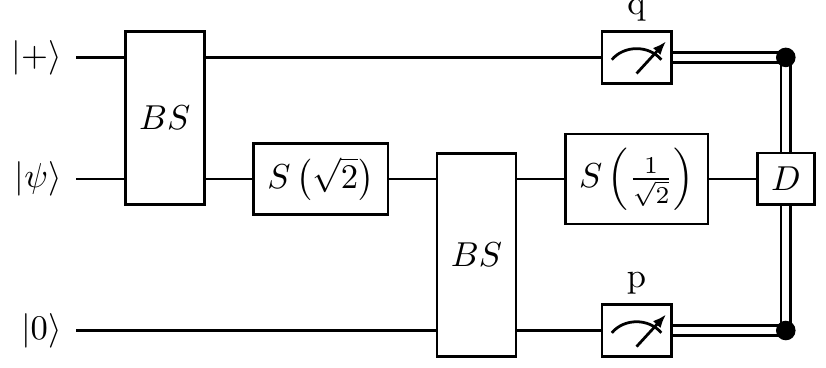}}\\
	\subfloat[]{\label{fig:improved_knill_glancy}\includegraphics{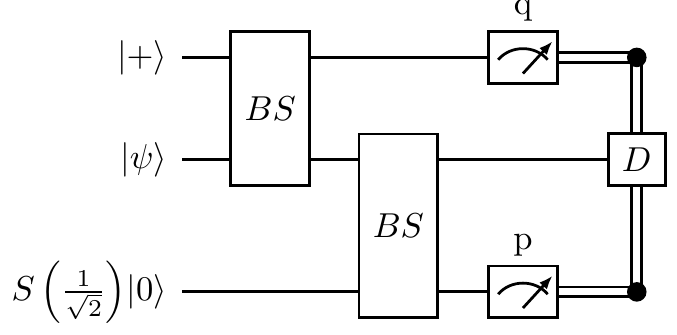}}
	\caption{Different methods to obtain the syndrome information of a square GKP code (a) Steane-inspired approach introduced in \cite{gkp}. The CNOT gates are implemented by CSUM gates where each can be decomposed into two beam splitters and two squeezers. (b) Knill-Glancy scheme \cite{Knill_Glancy} where each CSUM gate is replaced by a single beam splitter and squeezer (c) improved Knill-Glancy scheme where we only need beam splitters and an offline-squeezed state.}
\end{figure}

\section{Improvement of Syndrome Measurements}\label{sec:four}

In many works \cite{kosuke_minimal_measurement,PhysRevA.101.012316,PhysRevA.99.032344,PhysRevX.8.021054,PhysRevA.102.052408} concatenations of GKP codes with higher-level qubit codes are considered and the syndrome measurements of the GKP code and the high-level code are done independently.
This means one first obtains the GKP syndrome information for correcting the small shifts and then one obtains the syndrome information of the higher level code for correcting the larger shifts.
Each of these measurements typically make use of a GKP-like ancilla state which is costly.
Therefore, we discuss alternative measurement schemes which only make use of a minimal number of measurements.

Let us begin with the qubit case where we concatenate an $n$-mode GKP code with an arbitrary stabilizer code. We show that by using $2n$ measurements we not only obtain the GKP syndrome information of the $n$-mode GKP code, but also the additional syndrome information for decoding the higher level code.
This can be seen rather easily by describing the whole concatenated code by a set of independent (Weyl-Heisenberg) stabilizer generators. 
The stabilizer of the GKP code can be obtained by applying logical Pauli operators twice. In an naive approach one would construct a set of stabilizer generators by first considering the stabilizers of the GKP code and then adding the qubit stabilizers. However, these stabilizer generators are not independent, because we can apply the qubit-like stabilizers twice in order to obtain stabilizer generators of the GKP code. Thus we can remove these, such that we still have $2n$ independent stabilizer generators. When encoding quantum information into a code we have a product state of ancillas in Pauli eigenstates. This state can therefore be described by $2n$ independent stabilizer generators. In order to do the encoding we perform a sequence of Clifford (Gaussian) operations, changing the actual stabilizer generators but their number remains invariant. Thus, we only need to measure the $2n$ independent stabilizer generators in order to obtain full syndrome information. Furthermore, we can generalize this result to arbitrary qudit dimensions $D$ by using a different proof technique based on lattice theory instead due to technical difficulties. The proof is given in App. \ref{app:minimal set}. This result is quite remarkable, because one needs no additional measurements in order to obtain the additional syndrome information of the higher level code, which consists of up to $2n$ (Pauli) stabilizer generators for the case of non-prime $D$. While such minimal measurements have been proposed in an ad-hoc way for some codes \cite{kosuke_minimal_measurement} \footnote{In Ref. \cite{kosuke_minimal_measurement} the authors proposed to measure the stabilizers $\exp(i\sqrt{\pi}(\hat{q}_1+\hat{q}_2)),\exp(i\sqrt{\pi}(\hat{q}_2+\hat{q}_3)),\exp(i2\sqrt{\pi}\hat{q}_3)$}, in the next sections we discuss two schemes which allow us to obtain the full syndrome information in a systematic way for general GKP codes concatenated with stabilizer codes employing only GKP-like states, linear optics and homodyne measurements. 

\subsection{Teleportation-based error correction}
\subsubsection{GKP syndrome}

Here we will discuss how to obtain the syndrome information of a general GKP code which will be a building block for the scheme that additionally also gives the syndrome information of the high-level code.
Let us first discuss the special case of a GKP qudit code using a square lattice and show that it is possible to obtain the syndrome information without needing in-line squeezing operations. 
Recall that the (Weyl-Heisenberg) stabilizers of such a code encoding a qudit (with dimension $D$) are given by
\begin{equation}
\exp\left(-i \hat{p} \sqrt{2\pi D}\right) \text{ and } \exp\left(i \hat{q} \sqrt{2\pi D}\right)\,,
\end{equation} 
where $\hat{q}$ and $\hat{p}$ are quadrature operators of the harmonic oscillator.
The logical Pauli operators of the GKP code are given by
\begin{align}
\overline{X}=\exp\left(-i \hat{p}\sqrt{\frac{2\pi}{D}}\right) \text{ and } \overline{Z}=\exp\left(i \hat{q}\sqrt{\frac{2\pi}{D}}\right)\,.
\end{align}
Therefore, the logical information encoded in $\ket{\psi}^{GKP}$ is encoded in modular quadrature operators.
Let us consider a qudit Bell state 
\begin{equation}
\ket{\Phi_{00}}_{2,3}:=\frac{1}{\sqrt{D}}\sum_{k=0}^{D-1}\ket{k,k}_{2,3}\,,
\end{equation}
which can also be described by the two (qudit) stabilizers $X_2X_3$ and $Z_2 Z_3^{-1}$ \footnote{Even for non-prime $D$ there are no additional stabilizers needed.}.
We can construct all other Bell states via
\begin{equation}
\ket{\Phi_{rs}}_{2,3}:=\overline{X}^r_2 \overline{Z}^s_2\ket{\Phi_{00}}_{2,3}\,,
\end{equation}
where $r,s \in \mathbb{Z}_D$.
If we have such a qudit Bell state encoded in two GKP qudits, the Bell state stabilizer conditions are equivalent to
\begin{align}
\left(\hat{p}_2+\hat{p}_3 -s \sqrt{\frac{2\pi}{D}} \mod \sqrt{2\pi D}\right)\ket{\Phi_{rs}}_{2,3}^{GKP}=0\,,\\
\left(\hat{q}_2-\hat{q}_3-r \sqrt{\frac{2\pi}{D}} \mod \sqrt{2\pi D}\right)\ket{\Phi_{rs}}_{2,3}^{GKP}=0\,.
\label{eq:GKP_bell_stabilizer}
\end{align}
Notice that these two stabilizers alone do not define a GKP Bell state uniquely, because for example an infinitely squeezed two-mode squeezed state also satisfies these conditions.
\begin{figure}
\includegraphics[width=0.4\textwidth]{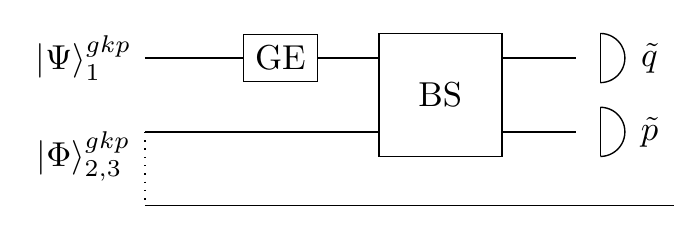}
\caption{A logical qudit is encoded in mode 1 and is affected by Gaussian errors ("GE"). Then it is coupled with one half of a logical Bell state pair via a balanced beam splitter. The position and momentum quadratures of the beam splitter output are measured. We can use these measurement results for error correction of the GKP code and for correcting the higher level code. The teleportation protocol actually also involves applying conditional displacements. However, when considering multiple rounds of this teleportation protocol we actually do not need to apply the displacement in every step, but we can keep track of the displacements and apply only one displacement in the end, because they only shift the measurement results of the next error correction cycle. This is similar to the Pauli-frame for qubits. The dotted line denotes that modes 2 and 3 share an entangled state.} 
\label{fig:teleportation_scheme}
\end{figure}

We consider a beam splitter with the transformations
\begin{align}
\hat{\tilde{q}}=\frac{1}{\sqrt{2}}\left(\hat{q}_1-\hat{q}_2\right)\,, \label{eq:q_bs}\\
\hat{\tilde{p}}=\frac{1}{\sqrt{2}}\left(\hat{p}_1+\hat{p}_2\right)\,.\label{eq:p_bs}
\end{align}
Let us first assume an ideal ancilla state $\ket{\Phi}_{2,3}^{GKP}$ and also an arbitrary ideal data state $\ket{\psi}_1^{GKP}$.
Now we first show that we can use the circuit illustrated in Fig.  \ref{fig:teleportation_scheme} for teleporting the information encoded in the GKP qudit:
\begin{align}
\nonumber&\quad\hat{q}_1 \mod \sqrt{2\pi D}\ket{\psi}_1^{GKP}\ket{\Phi_{rs}}_{2,3}^{GKP}\\\nonumber&= \hat{q}_1 -\hat{q}_3+\hat{q}_3\mod \sqrt{2\pi D}\ket{\psi}_1^{GKP}\ket{\Phi_{rs}}_{2,3}^{GKP}\\
\nonumber&=\hat{q}_1 -\hat{q}_2+r \sqrt{\frac{2\pi}{D}}+\hat{q}_3\mod \sqrt{2\pi D}\ket{\psi}_1^{GKP}\ket{\Phi_{rs}}_{2,3}^{GKP}\\&=\hat{q}_3+\sqrt{2}\hat{\tilde{q}}+r \sqrt{\frac{2\pi}{D}} \mod \sqrt{2\pi D}\ket{\psi}_1^{GKP}\ket{\Phi_{rs}}_{2,3}^{GKP}\,.
\end{align}
Here we only used the stabilizer property of the GKP Bell state. If we measure $\hat{\tilde{q}}$ and shift $\hat{q}_3$ by $\sqrt{2}\tilde{q}+r \sqrt{\frac{2\pi}{D}}$, we then have successfully teleported the information encoded in the modular position quadrature. Similarly we can teleport the information encoded in the modular momentum quadrature by measuring $\hat{\tilde{p}}$ and shifting $\hat{p}_3$ accordingly:
\begin{align}
\nonumber&\quad \hat{p}_1 \mod \sqrt{2\pi D}\ket{\psi}_1^{GKP}\ket{\Phi_{rs}}_{2,3}^{GKP}\\ \nonumber&= \hat{p}_1 -\hat{p}_3+\hat{p}_3\mod \sqrt{2\pi D}\ket{\psi}_1^{GKP}\ket{\Phi_{rs}}_{2,3}^{GKP}\\ \nonumber
&= \hat{p}_1+\hat{p}_2+\hat{p}_3-s\sqrt{\frac{2\pi}{D}}\mod \sqrt{2\pi D}\ket{\psi}_1^{GKP}\ket{\Phi}_{2,3}^{GKP}\\&= \hat{p}_3+\sqrt{2}\hat{\tilde{p}}-s\sqrt{\frac{2\pi}{D}} \mod \sqrt{2\pi D}\ket{\psi}_1^{GKP}\ket{\Phi}_{2,3}^{GKP}\,.
\end{align}
The demonstrated teleportation is exactly the well-known qudit-teleportation applied to GKP qudits, if we assume that the GKP states are in their codespace such that they are well defined qudits. We already saw that we can use the measurement result from the two homodyne detections for shifting the GKP states back into the codespace. Thus we can understand the protocol in the following way: First we use the homodyne measurement for correcting small shifts to the nearest codeword in mode 1 and then we perform a common qudit teleportation protocol, teleporting the encoded information into mode 3. 
Therefore, the only actually interesting observation lies in the fact that the homodyne measurements give us information about the measured GKP Bell state and the GKP syndrome information at the same time. Also notice that the displacement for correcting the small shift together with the displacement from the teleportation protocol reduces to a single GKP Pauli-operation.
\subsubsection{Incoherent noise}
Up to now, we considered only ideal GKP states which are clearly unphysical since they are not normalizable and have infinite energy. Realizable approximate GKP states are for example given by a  coherent superposition of Gaussian displacements acting on an ideal GKP state. For simplicity we will consider an error model of finite squeezing where we replace the coherent displacements by stochastic ones. 

 First we will show that we can correct Gaussian shift errors acting on the data mode, while assuming noiseless ancilla states. Later we show that we can also consider noisy ancilla states (in our error model) and this is equivalent to considering noiseless ancilla states, but more noise on the data mode.

In order to perform error correction of the GKP code we actually have to measure $\hat{q} \mod \sqrt{\frac{2\pi}{D}}$ and $\hat{p} \mod \sqrt{\frac{2\pi}{D}}$ which give the result '0' for square-lattice GKP codewords. For correcting shift errors we simply apply the smallest shift needed to obtain a codeword again:
\begin{align}
\label{eq:gkperrorsyndrome}
\nonumber&\sqrt{2}\hat{\tilde{q}} \mod \sqrt{\frac{2\pi}{D}}\ket{\psi}_1^{GKP}\ket{\Phi_{rs}}_{2,3}^{GKP}\\&=\hat{q}_1-\hat{q}_2 \mod \sqrt{\frac{2\pi}{D}}\ket{\psi}_1^{GKP}\ket{\Phi_{rs}}_{2,3}^{GKP}\nonumber \\&= \hat{q}_1 \mod \sqrt{\frac{2\pi}{D}}\ket{\psi}_1^{GKP}\ket{\Phi_{rs}}_{2,3}^{GKP}\,.
\end{align}
For the last step we used our assumption that mode 2 is part of a perfect GKP state and thus $\hat{q}_2 \mod \sqrt{\frac{2\pi}{D}}=0\,.$
Hence, we know the syndrome information and can apply the corresponding correction shift onto mode 3. When we consider the shift from the teleportation and the correction shift together, we obtain simply a Pauli operator. The same reasoning holds for the modular momentum quadrature.

Let us now consider also noisy ancilla states (assuming a random shift model). Let $v_i$ denote the random variable describing the momentum shift acting on mode $i$ and $u_i$ denote the corresponding 
random variable for the position shifts. As it can be seen in Eq. (\ref{eq:p_bs}) a shift of $\hat{p}_1$ by $v_1$ and a shift of $\hat{p}_2$ by $v_2$ have the same outcome of the measurement as a shift of $\hat{p}_1$ by $v_1+v_2$. Similarly one can show by using Eq.(\ref{eq:q_bs}) that the position shifts acting on modes 1 and 2 have the same effect on the measurement outcome as a shift of $\hat{q}_1$ by $u_1-u_2$.
  We interpret the shift errors on mode 2 as additional shifts on mode 1 and the shifts of mode 3 are the finite squeezing shifts of the data GKP qudit in the next error correction step. Also notice that there is no need (in the random shift model) to perform the displacement operations after each correction step, but one can keep track of them similar to a Pauli frame.
We did not assume a particular distribution of the random variables describing the shift errors and their possible correlations. We will do this later when we discuss different approaches of generating GKP Bell states.

\subsubsection{General GKP codes}

Let us now generalize this scheme from a GKP code based on a square lattice to general GKP codes which may even be defined on $n$ modes. The stabilizer generators span a lattice in the $2n$-dimensional phase space.
The corresponding logical Pauli operators are of the form $\overline{X}_j=\exp(-i a_j \hat{P}_j)$ and $\overline{Z}_j=\exp(i b_j \hat{Q}_j)$ where $\hat{Q}$ and $\hat{P}$ are linear combinations of quadrature operators, fulfilling the canonical commutation relation $[\hat{q}_k,\hat{p}_l]=i\delta_{kl}$, and some $a_j,b_j\in\mathbb{R}$. Since we are considering quantum teleportation, our resource states must be Bell states encoded in the same code as the input mode. For measuring the Bell states we only need to measure $\overline{X}_{j,1}\overline{X}_{j,2}=\exp(-i a_j (\hat{P}_{j,1}+\hat{P}_{j,2}))$ and $\overline{Z}_{j,1}\overline{Z}_{j,2}^{-1}=\exp(-i b (\hat{Q}_{j,1}-\hat{Q}_{j,2}))$.  However, the observables $\hat{P}_{j,1}+\hat{P}_{j,2}$ and $\hat{Q}_{j,1}-\hat{Q}_{j,2}$ commute such that we can measure them simultaneously instead of only measuring the quantities modulo some constant. We have shown that it is possible to interpret mode 2 as noiseless when considering more noise on mode 1. Measuring the relevant syndrome means measuring $\hat{Q}_{j,1} \text{mod} \frac{2\pi}{D a_j}$ and $\hat{P}_{j,1} \text{mod} \frac{2\pi}{D b_j}$. However, we know that the state in mode 2 is part of the logical Bell state and therefore the relevant modulos of mode 2's quadrature operators are zero. Thus, we can obtain the modulo of the quadrature operators of mode 1 by applying the mod function on the measurement outcome of the commuting observables $\hat{Q}_{j,1}-\hat{Q}_{j,2}$ and $\hat{P}_{j,1}+\hat{P}_{j,2}$. Recall that $\hat{P}$ and $\hat{Q}$ are linear combinations of quadrature operators and therefore we can measure them with passive, linear optics and homodyne measurements. 

Let us first explain why this is possible in the single-mode case. In order to measure $\hat{Q}_{1}-\hat{Q}_{2}$ and $\hat{P}_{1}+\hat{P}_{2}$ we have to couple modes 1 and 2 with a 50:50 beam splitter and then we need to measure the resulting operators $\hat{\tilde{Q}}_2$ and $\hat{\tilde{P}}_1$ which are both linear combinations of position and momentum operators. Equivalently, it is possible to represent this linear combination in polar coordinates $\alpha \hat{q}+\beta \hat{p}=\gamma \left(\cos(\theta)\hat{q}+\sin(\theta)\hat{p}\right)$ with $\alpha,\beta,\gamma,\theta\in \mathbb{R}$. Thus, the measurement of the linear combination can be understood as the measurement of a rotated quadrature which was squeezed in the direction of $\theta$ where the squeezing corresponds to the factor $\gamma$. However, we can also understand the measurement of the linear combination as a measurement of the rotated quadrature operator where we classically rescale the measurement outcome by a factor $\gamma$. In other words we have replaced the squeezing operation by multiplication in a post-processing step of the measurement data. Let us now discuss the general multi-mode case ($n\geq1$). We need to measure all $\hat{\tilde{P}}_{j}$ and $\hat{\tilde{Q}}_{j}$ ($j\in\{1,\dots,n\}$). Here, we only consider the case of $\hat{\tilde{P}}_{j}$, because the other one works analogously. 
In the symplectic representation $\tilde{P}_j$ of the operators  $\hat{\tilde{P}}_j$, we see that $\text{span}_\mathbb{R}(\tilde{P}_1,\dots,\tilde{P}_n)$ generates a $n$-dimensional linear subspace of the phase space. 
However, notice that the basis $\{\tilde{P}_1,\dots,\tilde{P}_n\}$ does not necessarily form an orthonormal basis. Let $\{\tilde{\xi}_1,\dots,\tilde{\xi}_n\}$ be an orthonormal basis of the same linear subspace. 
Then there exists an invertible ($n\times n$) matrix $A$ relating both bases via
\begin{align}
	\tilde{P}_j=\sum_{i=1}^nA_{ji}\tilde{\xi}_i.
\end{align}
Thus, we can implement a measurement of $(\hat{\tilde{p}}_j,\dots,\hat{\tilde{p}}_2)$ by measuring  $(\hat{\tilde{\xi}}_1,\dots,\hat{\tilde{\xi}}_n)$ and applying the matrix A onto the classical measurement data. Since $\{\tilde{\xi}_1,\dots,\tilde{\xi}_n\}$ is an orthonormal basis, we can employ linear optical transformations, which induce arbitrary orthogonal transformations on this $n$-dimensional linear subspace (symplectic, orthogonal transformations on the whole $2n$-dimensional phase space), and quadrature measurements of independent modes in order to measure $\{\tilde{\xi}_1,\dots,\tilde{\xi}_n\}$.

   Therefore, for measuring the syndrome of any GKP code we only need offline-squeezing operations and all inline operations are passive, linear optics and homodyne measurements. This result is not obvious, because initially we only knew that it is possible for the square lattice-GKP code.
   A straightforward way of showing this generalization would be by going from a general lattice to a square one, performing the syndrome measurement and going back to the general lattice. The transformation between two GKP codes is realized by a Gaussian operation, which in general involves squeezing operations, thus the resulting circuit for performing the syndrome measurement is given by a linear optical operation conjugated by a Gaussian one. However, after conjugation we do not necessarily obtain a linear optical operation (for a single-mode counter-example consider e.g. a $\frac{\pi}{2}$ phase-shift conjugated by a squeezing operation).

\subsection{Obtaining the higher-level syndrome}
Let us furthermore not only consider GKP qudit codes, but a concatenation with a high-level $[n,k,d]_D$ stabilizer code built with GKP qudits. 
Here, in order to obtain the syndrome of the high-level code we explicitly perform Knill's error correction by teleportation scheme \cite{Knill_teleportation}.
The qudit-teleportation in the Knill scheme is here given by the GKP-teleportation discussed previously, which is also capable of additionally obtaining the syndrome of the GKP-code provided the resource-state is a GKP-Bell state. A logical Bell state is given by a superposition of GKP-Bell states, because the set of Bell states forms an orthonormal basis for two qudits. As it can be seen in Eq. \ref{eq:gkperrorsyndrome}, we can obtain the error syndrome of the GKP code for any GKP Bell state and therefore by linearity also when using the logical Bell state. We can then use this syndrome information for mapping mode 1 into the code space (via classical post-processing) and then we can correct errors of the high-level code simply by applying Knill's error correction by teleportation protocol and treating each of the three modes as a qudit.

As a consequence, this scheme demonstrates that on the one hand one does not need to measure the $2n$ stabilizers in order to obtain the syndrome of the individual GKP qudits followed by an additional measurement of the high-level code's stabilizer, but $2n$ measurements suffice and on the other hand inline squeezing is neither needed for correcting small shifts on GKP qudits nor for obtaining the high-level error syndrome. In the original paper \cite{Knill_teleportation} it was shown that the Knill scheme works for arbitrary qubit codes. In App. \ref{app:knill} we show that it also works for qudit CSS codes with an arbitrary qudit dimension $D$. 
Furthermore, for non-CSS codes one can find a similar scheme where we need an ancilla state different from a logical Bell state. This difference comes from the asymmetry in the stabilizers $Z_1 Z_2^{-1}$ and $X_1 X_2$ of a qudit Bell state. $X$ and $Z$ are treated differently in the general qudit case, while in the special case of qubits, $X$ and $Z$ are treated equally, because the Pauli operators are self-inverse. 

 \subsection{Example: three-mode code}
As an example let us consider the error correction of the concatenation of square GKP qubits with the three-qubit bit-flip code. It was already shown in Ref. \cite{kosuke_minimal_measurement} that the codes performance can be improved significantly by using the complete (analog) error syndrome of the GKP syndrome measurement in the decoder of the high-level stabilizer code. This means we assign a value of reliability to every single GKP error correction, i.e. the further we are away from a codeword the lower the value of reliability.  As it can be seen in Fig \ref{fig:3modeexample} we perform error correction by coupling the (three-mode) input state with one half of an ancilla state consisting of a Bell state encoded in two three-qubit codes with transversal 50:50 beam splitters. Then we perform homodyne measurements on the first six modes which allow us to calculate the needed correction shift as the six measurements contain the same information as the measurement of the code's six stabilizer generators (explicit stabilizers are given in App. \ref{app:analoginfo}).
\begin{figure}
	\centering
	\includegraphics[width=0.7\linewidth]{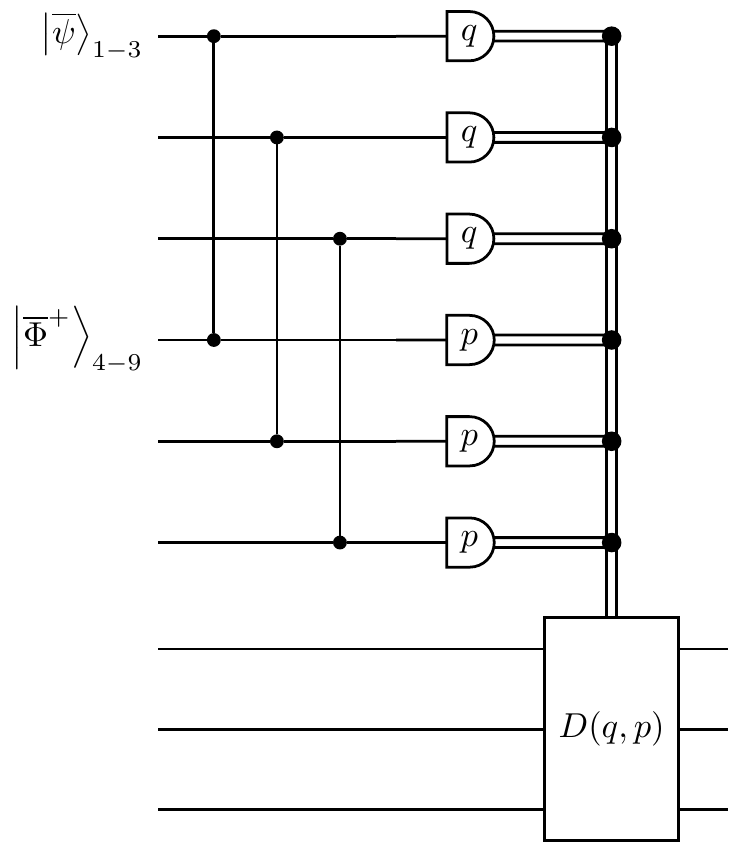}
	\caption{Error correction by teleportation for the concatenation of a square GKP code with a three-qubit repetition code. In the first three modes we have the noisy input encoded in the code. In modes 4-9 we have an encoded Bell state of the full code where we then couple the first half with the three input modes transversally with beam splitters. We then perform homodyne measurements and apply conditional displacements on the second half of the encoded Bell state. }
	\label{fig:3modeexample}
\end{figure}

When we consider ideal codes followed by i.i.d. Gaussian noise, we can calculate the resulting error channel we obtain when using the analog information in an exact approach instead of requiring simulations as in \cite{kosuke_minimal_measurement,analog_gkp_repeater}. The crucial observation allowing this is that the concatenation of square GKP codes with a stabilizer code is a code with a more sophisticated lattice in the phase space. The exact calculation can be performed by calculating the voronoi cells for $\mathcal{L}^\bot/\mathcal{L}$. More details can be found in App. \ref{app:analoginfo}. 

\subsection{Linear-optics state generation}
\subsubsection{GKP Bell states}
Up to now we have not assumed anything about the random variables despite being Gaussian. However, depending on the actual state generation there might be correlations involved. For example let us consider the case where we generate a square-GKP Bell state by coupling a noisy $\ket{+}_2$ and a noisy $\ket{0}_3$ with a CSUM gate. We further assume that the noise of both GKP qubits consists of independent, identically distributed  (i.i.d.) Gaussian shifts in position ($u_2^*,u_3^*$)  and momentum ($v_2^*,v_3^*$) with variance $\Delta^2$. Due to the CSUM gate we see that the random variables $u_2=u_2^*, v_2=v_2^*-v_3^*$ and $u_3=u_2^*+u_3^*,v_3=v_3^*$ contain some correlations. The states of modes 2 and 3 are used in different error correction steps and in usual decoding schemes (quite recently decoders making use of the syndrome information of multiple rounds have been considered \cite{gkp_offline})  it is assumed that each correction step only uses local information, neglecting the correlations. Therefore, it seems that the CSUM gate amplifies the noise such that we have momentum shifts with variance $2\Delta^2$ in mode 2 and position shifts of variance $2\Delta^2$ in mode 3. When additionally considering the noise from mode 1 we obtain the same result as in Ref. \cite{Knill_Glancy} that the sum of initially three random variables of individual variance $\Delta^2$ should be smaller than $\sqrt{\frac{2\pi}{D} }/2$. Thus, in terms of thresholds we do not gain anything by using a teleportation scheme instead of the Knill-Glancy scheme.

 Let us now consider a different scheme for generating Bell states as introduced in Ref. \cite{bosonic_gate_teleportation} using only a beam splitter to couple two noisy GKP-like states. Thanks to the simple linear optical coupling the resulting random variables $u_2,u_3,v_2,v_3$ are all i.i.d. with variance $\Delta^2$. This allows us to use simple decoders depending only on the syndrome information from this correction step without loosing the capability of correcting errors. Thus, in an error correction by teleportation we only need to consider $2\sigma_{sq}^2$ using this beams splitter approach instead of $3\sigma_{sq}^2$ when using CSUMs for generating the Bell states and neglecting correlations between different teleportation steps.
 
 In Ref. \cite{bosonic_gate_teleportation}, it was shown for a square GKP qubit code that a Bell state can be obtained by mixing two 'qunaught' states at a 50:50 beam splitter by using the state picture. Here, we will first reproduce this result in the stabilizer formalism, such that it will be easy to generalize the result to more general GKP codes.
 
 Now we will consider the slightly more general case of a square lattice GKP code with even qudit dimension $D$.
 Consider the two single-mode states described by the stabilizer group generated by the set of stabilizer generators 
 \begin{align}
 	\label{fig:inputstate}
 	\nonumber\{\exp(i\sqrt{\pi D}\hat{q}_1),\exp(i\sqrt{\frac{4\pi}{D}}\hat{p}_1),\\\exp(i\sqrt{\frac{4\pi}{D}}\hat{q}_2),\exp(i\sqrt{\pi D}\hat{p}_2)\}\,.
 \end{align}
 Let us apply a 50:50 beam splitter mixing both modes, resulting in the stabilizer generators
 \begin{align}
 \nonumber\{	\exp(i\sqrt{\frac{\pi D}{2}}(\hat{q}_1+\hat{q}_2)),\exp(i\sqrt{\frac{2\pi}{D}}(\hat{p}_1+\hat{p}_2))\,,\\
 	 	\exp(i\sqrt{\frac{2\pi}{D}}(\hat{q}_1-\hat{q}_2)),\exp(i\sqrt{\frac{\pi D}{2}}(\hat{p}_1-\hat{p}_2))\}\,.
 \end{align}
 This set of stabilizer generators already describes the canonical Bell state of the square GKP-code. However, we will consider a different set of stabilizer generators  by multiplying stabilizers such that it is more obvious that this set stabilizes the Bell state:
 \begin{align}
 	\label{fig:outputstate}
 \nonumber\{	\exp(i\sqrt{2\pi D}\hat{q}_1),\exp(i\sqrt{\frac{2\pi}{D}}(\hat{p}_1+\hat{p}_2))\,,\\
	\exp(i\sqrt{\frac{2\pi}{D}}(\hat{q}_1-\hat{q}_2)),\exp(i\sqrt{2\pi D}\hat{p}_1)\}\,.
 \end{align} 
 Here we multiplied the 1st (4th) stabilizer generator $D/2$ times with the 3rd (2nd) stabilizer generator. Since the number of multiplications must be an integer, we have the restriction that $D$ has to be even.
   For odd $D$ it seems that no scheme using only linear optics and two product states is possible (a simple beam splitter solution does not exist), but we have no rigorous proof for this. We obtained our results ($n=1$) by going through the above steps in opposite direction in order to obtain the input state. We started with the stabilizers of the desired Bell state (Eq. \ref{fig:outputstate}), applied an arbitrary two-mode passive linear-optical transformation and tried to multiply stabilizers such that there are only local stabilizer pairs for each mode (Eq. \ref{fig:inputstate}). Notice that this arbitrary operation can be decomposed into a relative phase followed by a beam splitter followed by another relative phase and a global phase. The two phases applied after the beam splitter are single-mode operations and are therefore useless for changing entanglement, so that we can ignore them. 
 
 Furthermore, it is also possible to show similar results (for even $D$) not only for the square lattice GKP code, but for more general ones. However, this is meant in the sense that we can obtain a $2n$-mode Bell state by mixing two $n$-mode states at $n$ 50:50 beam splitters in a transversal fashion. The proof for this is given in appendix \ref{app:bell_decomp}.
 
 \subsubsection{Higher encoded GKP Bell states}
 The most important ingredient for the error correction by teleportation of the high level code is the generation of the logical Bell state.
  Here, we discuss the possibility of generating these high level states by sending product states of single-mode grid states through a linear optical network.
   Such a generation would be nice for two reasons. First, the linear optical operations do not amplify the noise (we assume the initial noise is isotropic) and second inline squeezing is experimentally demanding and usually implemented via the teleportation of a finitely squeezed state necessarily introducing errors due the finite squeezing.
 It is easy to see that this linear optical network is unable to transform small GKP codes and states into a concatenation of a GKP code with a stabilizer code, because linear optical operations are represented by symplectic, orthogonal matrices in phase space and due to the orthogonality the code distance remains invariant (details are given in App. \ref{app:codedistance_preservation}). 
 However, while this shows that it is impossible to encode arbitrary quantum information into a code of higher code distance using linear-optical transformations, it might still be possible to generate some codewords which can then be used for performing error correction.

As the next step we discuss this loophole for relevant cases.  Remember that linear-optical transformations are represented by orthogonal and symplectic linear maps in the phase space representation. We will now use the orthogonality in order to obtain necessary conditions.
 Thus, we need to check whether the desired state admits a lattice representation with an orthogonal basis. Conditions for the existence of an orthogonal basis are discussed in Ref. \cite{orthogonal_lattice} for so-called construction-A lattices (for every linear code $C\in \mathbb{Z}_p^n$ we can construct a lattice $\{x\in\mathbb{Z}^n|x \mod p\in C\}$), which appear when we consider the concatenation of a  GKP code with a high level CSS code (the codewords of $C$ correspond to the stabilizers of the high level code, while the mod corresponds to the stabilizers of the low level GKP code), where the stabilizers of this concatenation are given by the columns of the matrix 
\begin{align}
	A=\frac{1}{\sqrt{D}}\mathds{1}_{2n\times 2n}\cdot\begin{pmatrix}
	G^X&0\\
	0&G^Z
	\end{pmatrix}\,.
\end{align} Each column of $G^X$, $G^Z$ corresponds to a basis element of the corresponding construction-A lattice and each column of $\frac{1}{\sqrt{D}}\mathds{1}_{2n\times 2n}$ gives the phase space representation of the $X$ and $Z$ operators of the square lattice GKP code.
 Because much experimental effort has been made in order to generate rectangular grid states \cite{Fluehmann2019,Campagne-Ibarcq2020}, it is a relevant question whether these states can be transformed into codewords of the concatenation of the square lattice GKP code with a CSS code by passive linear-optical operations.
  Thus, we want $A=U\cdot A'$ to hold where $U$ is an orthogonal, symplectic matrix describing the passive transformation  and $A'$ is a diagonal matrix denoting the stabilizers of independent rectangular grid states.
   Since $A'$ and $U$ are orthogonal matrices, it is necessary that $A$ and therefore also $G^X$ and $G^Z$ (needs to hold for at least one basis) have to be orthogonal matrices in order for a passive transformation to exist.

Before we consider a large class of CSS codes let us first consider a specific example, namely the three-qubit GKP-GHZ state. Its qubit stabilizer generators are $X_1X_2X_3, Z_1Z_2$ and $Z_2Z_3$. As a consequence we obtain

\begin{align}
	G^X=\begin{pmatrix}
	1 & 0 & 0 \\
	1 & 2 & 0 \\
	1 & 0 & 2
	\end{pmatrix} \,,\qquad
	G^Z=\begin{pmatrix}
	1 &0 &0 \\
	1 & 1 & 0 \\
	0 & 1 & 2
	\end{pmatrix}\,,
\end{align} as a possible basis of the construction-A lattices generated by the code $C=C_1\oplus C_2=\text{span}_{\mathbb{Z}^2}(0,0,0|1,1,1)\oplus \text{span}_{\mathbb{Z}^2}((1,1,0|0,0,0),(0,1,1|0,0,0))$. Since $C_1$ has code distance 3, it is obvious that code $C$ cannot be factored into (permutated) linear subcodes of maximum length 2. Hence, by Ref. \cite{orthogonal_lattice} there exists no orthogonal basis and thus we are not able to generate the GHZ state from single-mode grid states and linear optics.

 In CSS codes the set of $X$-type operators (involving stabilizers and logical operators) corresponds to codewords of $C_Z$ and the set of $Z$-type operators corresponds to codewords of $C_X$. Therefore, all operators stabilizing a logical Pauli-eigenstate correspond to a subcode  of $C_Z\oplus C_X$ using the symplectic representation and its code distance $d(C_Z\oplus C_X)$ is given by $\min\left(d(C_Z),d(C_X)\right)$. We are mostly interested in codes which are able to correct at least arbitrary single qubit errors demanding that the minimum code distance is at least 3.
In Ref. \cite{orthogonal_lattice} it was shown that a construction-A lattice over a binary field can only have an orthogonal basis if the corresponding code can be decomposed in a specific structure with a code distance of at most 2. Thus it is impossible in the qubit case to find such a passive transformation. In the qutrit case we can make a similar argument where the code distance must not be greater than 3 (it might still be impossible for 3), i.e. we can exclude the possibility of a passive transformation for high-distance codes. Up to now we only considered the concatenation with a square lattice GKP code, but in our argument we only used the property that the matrix representing the $X$- and $Z$-operators of the GKP code is orthogonal. Therefore, the result also holds for concatenations involving any GKP code fulfilling this relation. 

Up to now we assumed idealized infinitely squeezed GKP states in the proof of the above no-go statement, but a similar argument also works for the physically more relevant case of approximate GKP states with coherent Gaussian displacement errors where the Gaussian's covariance matrix needs to be proportional to the identity up to symplectic transformations.
We make use of the finite-squeezing stabilizers introduced in Ref \cite{Royer_2020}, where finite squeezing with coherent Gaussian displacement errors (covariance matrix proportional to the identity) are applied by the operator $e^{-\Delta \hat{n}}$ and this transforms the stabilizer of an ideal GKP state $\exp(i\hat{g})$ to \begin{equation}
e^{-\Delta \hat{n}}\exp(i\hat{g})e^{\Delta \hat{n}}=\exp(i\left(\hat{g}\cosh(\Delta^2)+i\hat{\tilde{g}}\sinh(\Delta^2)\right)) \,, \label{eq:fin_sq_stabilizer}
\end{equation} where $\hat{g}$ and $\hat{\tilde{g}}$ are (real) linear combinations of quadrature operators and $\hat{n}$ is the total photon number in all modes.  Using the (canonical extension of the) symplectic representation one can map the stabilizer conditions of the finitely squeezed  states to a lattice embedded in $\mathbb{C}^{2n}$ instead of $\mathbb{R}^{2n}$. Since Gaussian unitary operations do not couple the real and imaginary parts in the symplectic representation, the real part of the lattice also needs to fulfill the orthogonality constraints as for the ideal GKP states independent from the imaginary part.  
Up to scaling factors we have the same problem as in the infinite squeezing case and since scaling factors are irrelevant for orthogonality, we again obtain a no-go result. Let us now briefly show that this holds for all Gaussians with a covariance matrix which is related by a symplectic transformation A to a covariance matrix proportional to the identity. We can see this by first applying $\hat{A}^{-1}$ to the ideal desired state, followed by $e^{-\Delta \hat{n}}$ in order to introduce the isotropic Gaussian noise followed by $\hat{A}$ bringing the covariance matrix to the desired form.
The resulting stabilizer is then given by
\begin{align*}
	\exp(i\left(\hat{A}\hat{A}^{-1}\hat{g}\hat{A}\hat{A}^{-1} \cosh(\Delta^2)+i \hat{A} \hat{\widetilde{A^{-1}gA}}\hat{A^{-1}}\sinh(\Delta^2)\right))\\=
	\exp(i\left(\hat{g}\cosh(\Delta^2)+i  \hat{A} \hat{\widetilde{A^{-1}gA}}\hat{A^{-1}}\sinh(\Delta^2)\right))\,,
\end{align*} which has the exact form as in Eq. \ref{eq:fin_sq_stabilizer}.

As it is impossible to build logical Bell states of a high level GKP code from single-mode grid states with linear optical transformations, one might be wondering if one could use linear optical transformations and two suitable $n$-mode grid states as a resource instead. However, this also turns out to be impossible for simple transversal beam splitters (see App. \ref{app:bell_decomp}), although we have not proven yet the impossibility of this with general linear optics.

An alternative approach to get rid of inline squeezing operations circumventing this no-go was shown in Ref. \cite{xanadu_linearoptics} where the authors propose to generate a $n$-mode GKP-cluster state by applying a linear optical transformation on $4n$ rectangular GKP states, performing homodyne measurements on $3n$ modes and applying conditional displacements.
Thus one might think that one also obtains the advantage of amplifying no noise. While technically true, one adds additional noise due to the additional finitely squeezed GKP states. Strictly speaking this approach would introduce even more noise than the canonical circuit involving (ideal) CZ-gates, because by applying circuit identities one can show (see \cite[Fig. 2]{xanadu_linearoptics}) that the linear optical scheme is equivalent to the canonical scheme up to some CSUM-gates which act as the identity on the code space, but propagate noise from the auxiliary states to the data state. Another disadvantage of this scheme, despite its conceptual beauty and other possible practical advantages, lies in the overhead of the required costly GKP states.

It is an interesting question whether there exist similar schemes with a lower overhead, potentially introducing less noise than the canonical encoding scheme.

 \section{Knill-Glancy error correction}
 \label{sec:knill-glancy}
In the previous section we discussed one scheme allowing us to obtain the full error syndrome without using inline squeezing. In this section  we will consider another such scheme.
 This scheme is an improvement of the Knill-Glancy scheme such that all squeezing operations only act on ancilla states.
  For the square lattice qubit GKP code this improved scheme was already (independently from our work) proposed in Ref. \cite{gkp_offline}.
  
  Here we will first discuss the stabilizer formalism and measurements by discussing the error correction of one quadrature in the original Knill-Glancy scheme as an example. Then it is easy to first generalize the improved Knill-Glancy scheme to arbitrary $n$-mode GKP codes encoding qudits of arbitrary dimension $D$ (see App. \ref{app:knill_glancy_gkp}) and later we also show that we can obtain an analogous scheme in the case where we concatenate these general GKP codes with arbitrary CSS-codes (see App. \ref{app:knill_glancy_concat}).

 The stabilizers of the square qubit GKP code are $\exp\left(i2\sqrt{\pi}\hat{q}\right)$ and $\exp\left(i2\sqrt{\pi}\hat{p}\right)$. Let us first consider the correction of position shifts. Thus we have to consider a general GKP state and a GKP-$\ket{+}$ state. After the Gaussian error channel we have an (unknown) error operator $\exp(i\left(v_1\hat{q}_1+v_2\hat{q}_2-u_1\hat{p}_1-u_2\hat{p}_2\right))$. After this error the two-mode state is stabilized by the following four stabilizers:
 \begin{align*}
 \exp(-i v_1 2\sqrt{\pi})\exp\left(i2\sqrt{\pi} \hat{p}_1\right)\,,\\ \exp(-i u_1 2\sqrt{\pi})\exp\left(i2\sqrt{\pi} \hat{q}_1\right)\,,\\
 \exp(-i v_2 \sqrt{\pi})\exp\left(i\sqrt{\pi} \hat{p}_2\right) \,,\\ \exp(-i u_2 2\sqrt{\pi})\exp\left(i2\sqrt{\pi} \hat{q}_2\right)\,.
 \end{align*}
 After applying the beams splitter, we obtain the stabilizer generators:
 \begin{align*}
 \exp(-i v_12 \sqrt{\pi})\exp\left(i\sqrt{2\pi} (\hat{\tilde{p}}_1+\hat{\tilde{p}}_2)\right)\,,\\ \exp(-i u_1 2\sqrt{\pi})\exp\left(i\sqrt{2\pi} (\hat{\tilde{q}}_1+\hat{\tilde{q}}_2)\right)\,,\\
 \exp(-i v_2 \sqrt{\pi})\exp\left(i\sqrt{\frac{\pi}{2}} (\hat{\tilde{p}}_1-\hat{\tilde{p}}_2\right) \,,\\ \exp(-i u_2 2\sqrt{\pi})\exp\left(i\sqrt{2\pi} (\hat{\tilde{q}}_1-\hat{\tilde{q}}_2)\right)\,.
 \end{align*}
 As the next step we perform a position measurement of mode 2. We can then use the stabilizers to find the set of possible measurement outcomes. By multiplication we find that $\exp(-i2\sqrt{\pi}(u_1-u_2))\exp(i2\sqrt{2\pi}\hat{\tilde{q}}_2)$ is also a stabilizer and thus possible measurement values of $\tilde{q}_2$
 take the form of $\frac{u_1-u_2}{\sqrt{2}}+\frac{\sqrt{\pi}}{\sqrt{2}}z$ for $z\in \mathbb{Z}$.
 In order to obtain the stabilizers after the measurement we simply replace $\hat{\tilde{q}}_2$ by the measurement value of $\tilde{q}_2$. For the stabilizers involving $\hat{\tilde{p}}$ we simply take the smallest product of stabilizer generators such that there appears no $\hat{\tilde{p}}_2$. This is quite similar to the qubit stabilizer formalism, where one takes products of stabilizer generators such that there is only one stabilizer generator which anti-commutes with the observable.  Since we are not interested in the eigenstate after obtaining the measurement result we can discard this mode, such that we only need two stabilizer generators to specify our state. Thus the stabilizer generators are given by
 \begin{align*}
 \exp\left(-i2\sqrt{\pi}(v_1+v_2)\right)\exp\left(i2\sqrt{2\pi}\hat{\tilde{p}}_1\right)\,,\\
 (-1)^z\exp\left(-i\sqrt{\pi}(u_1+u_2)\right)\exp\left(i\sqrt{2\pi}\hat{\tilde{q}}_1\right)\,.
 \end{align*}
 It is now easy to check that after applying a squeezing operation (reducing the $q$ variances by a factor $\sqrt{2}$) and a position displacement by $\frac{\tilde{q}_2}{\sqrt{2}}-\frac{1}{2}\textrm{ mod}_{2\sqrt{\pi}}(2\sqrt{2}\tilde{q}_2)$ \footnote{This expression differs by a factor of -1 from Eq. 14 in Ref. \cite{Knill_Glancy} as we use a beam splitter with different phases.} we completed the error correction and are in a state which is stabilized by 
 \begin{align*}
 &\exp\left(-i2\sqrt{\pi}(v_1+v_2)\right)\exp\left(i2\sqrt{\pi}\hat{\tilde{p}}_1\right)\,,\\
 &\exp\left(-i2\sqrt{\pi}(u_1+\frac{1}{2}\textrm{ mod}_{2\sqrt{\pi}}(2u_2-2u_1))\right)\exp\left(i2\sqrt{\pi}\hat{\tilde{q}}_1\right)\,.
 \end{align*}
 However, this only shows that we are close to the code space of a GKP code, but we do not know if the information within the code space is disturbed. Therefore we have to check that up to small phases (corresponding to small errors remaining after the error correction) that we also have $\exp\left(i\sqrt{\pi}\hat{q}_1\right)\rightarrow \exp\left(i\sqrt{\frac{\pi}{2}}\hat{\tilde{q}}_1\right)$ which is easy to check (before applying the squeezing operation).
 However, in order to show $\exp\left(i\sqrt{\pi}\hat{p}_1\right)\rightarrow \exp\left(i\sqrt{2\pi}\hat{\tilde{p}}_1\right)$ we also need to exploit that the ancilla GKP qubit is in the $\ket{+}$ state, because otherwise we cannot have the product $\exp\left(i \sqrt{\frac{\pi}{2}}(\hat{\tilde{p}}_1-\hat{\tilde{p}}_2)\right)\exp\left(i \sqrt{\frac{\pi}{2}}(\hat{\tilde{p}}_1+\hat{\tilde{p}}_2)\right)=\exp\left(i\sqrt{2\pi}\hat{\tilde{p}}_1\right)$. When considering shift errors one simply has to check if the overall phase at the end is approximately '0' (no error) or '$\pi$' (error). Since we discarded stabilizer generators after the homodyne measurement it could be possible that we discarded too many such that we allow for too many states. However, after the measurement we only have one mode of interest, but still two independent stabilizer generators defining the code. Thus we did not discard too many stabilizers.
 
 In the improved Knill-Glancy scheme the first ancilla is still a $\ket{+}$ state, but the second ancilla is now a $\ket{0}$ state which is squeezed by a factor $\sqrt{2}$ which can already be incorporated in the state generation, while we do not use inline squeezing of the data mode (see Fig. \ref{fig:improved_knill_glancy}). For the case where we consider a concatenation with a CSS code we simply have to do the same and  replace the GKP Pauli-eigenstates by Pauli-eigenstates of the high-level code and all beam splitters and homodyne measurements are applied in a transversal manner.
 \section{Error propagation in stabilizer measurements}\label{sec:six}
 Let us consider prime qudit dimension $D$ and a high-level CSS code.
 Such a stabilizer code is also defined by $n-k$ stabilizer generators which generate the whole stabilizer group.
 Usually the syndrome of a stabilizer code is obtained by directly measuring the $n-k$ stabilizer generators.
 In order to measure the stabilizers, we couple an ancilla with the code's GKP qudits. The ancillas are finitely squeezed and therefore we need to carefully design our stabilizer measurements in such a way that a shift on one ancilla does not introduce errors in other stabilizer measurements. This has been done for the surface code in Ref \cite{PhysRevA.101.012316}. Here, we discuss whether this is possible for every CSS code and how these measurements need to be modified.

In this section we restrict ourselves to square lattice GKP codes concatenated with CSS codes.
 In order to perform stabilizer measurements of CSS codes one couples an ancilla state with the data qubits with controlled-$X$ ($\rm CX_{i,j}$) operations.
  For example, measuring the stabilizer $\prod_{i\in support} X_i$ can be realized by measuring the ancilla $a$ of $\prod_i \rm CX_{a,i}\ket{+}_a$ in the $X$ basis, while the stabilizer $\prod_{i\in support} Z_i$ can be measured by measuring $\prod_i\rm CX_{i,a}\ket{0}_a$ in the $Z$ basis.
  We implement the $\rm CX$-gate by using a CSUM-gate since we consider a square lattice GKP code.
    Notice that operators acting equally within the codespace do not necessarily act the same way outside of the codespace.
Furthermore, because ideal GKP states are unphysical, we are almost surely outside of the codespace and should therefore take these differences into account. 

When performing the $Z$-stabilizer measurements in a standard way the CSUM gates transfer momentum-shifts from the ancilla state originating from the finite squeezing to the data GKP states resulting in correlated momentum-shifts on multiple data GKP qudits. When performing the $X$-stabilizer measurements later these shifts may introduce errors in the syndrome. Especially due to the correlations these shifts can easily add up and overcome the threshold of correctable shifts as the variance of the sum of $n$ independent random variables increases linearly while the variance of $n$ times the same random variable increases quadratically.
Furthermore, due to the correlated shifts the faults of the stabilizer measurements would no longer be independent.

In Ref. \cite{PhysRevA.101.012316} the authors introduced a way of using CSUM and inverse CSUM gates exploiting the correlations of the shift errors such that they cancel in the next stabilizer measurement, and so there is no error propagation from one ancilla to another ancilla for the planar-square surface code. 

Let us now discuss this error propagation in a systematic way in an attempt to generalize the scheme from Ref. \cite{PhysRevA.101.012316} to more general quantum error-correcting codes with parameters $[n,k,d]_D$. Let us define the vector 
\begin{equation}
	\vec{\Delta}_\text{data}^T=\left(\vec{u}_d,\vec{v}_d\right)=\left(u_{d,1},\dots,u_{d,n},v_{d,1},\dots,v_{d,n}\right)
\end{equation}
 of random variables describing the shift errors ($u$ for position shifts and $v$ for momentum shifts) acting on data GKP qudits.
Similarly we can define such a vector for the ancilla GKP qudits which are used to measure the $X/Z$ stabilizer generators
\begin{align}
\vec{\Delta}_{X/Z}^T=\left(\vec{u}_{X/Z},\vec{v}_{X/Z}\right)=\left(u_{X/Z,1},\dots,u_{X/Z,l_{X/Z}},\right.\\
\left. v_{X/Z,l_{X/Z}},\dots,v_{X/Z,l_{X/Z}}\right)\,,
\end{align}
where $l_{X/Z}$ gives the number of $X$- or $Z$-type stabilizer generators. Suppose we assume that all data-GKP qudits performed their syndrome measurement before measuring the stabilizers of the higher code. This means that all $u$ and $v$ are i.i.d. Gaussian random variables with  mean 0 and variance $\sigma_{sq}^2$.

We now first perform the $X$-stabilizer measurements and due to the coupling we obtain the following error vectors
\begin{align}
	\vec{u}_d'&=\vec{u}_d+H_X^T\vec{u}_X\,,\\
	\vec{v}_d'&=\vec{v}_d\,,\\
	\vec{u}_X'&=\vec{u}_X\,,\\
	\vec{v}_X'&=\vec{v}_X-H_X \vec{v}_d\,.
\end{align}
In order to measure the $X$-stabilizer we measure the momentum quadrature of the ancillas and therefore we always obtain a faulty syndrome whenever a random variable in $\vec{v}_X'$ lies in the set of uncorrectable errors.

When we now perform the $Z$-stabilizer measurements we obtain due to the coupling the error vectors
\begin{align}
\vec{u}_d''&=\vec{u}_d'=\vec{u}_d+H_X^T\vec{u}_X\,,\\
\vec{v}_d''&=\vec{v}_d'-H_Z^T\vec{v}_Z=\vec{v}_d-H_Z^T\vec{v}_Z\,,\\
\vec{u}_Z'&=\vec{u}_Z+H_Z\vec{u}_d'=\vec{u}_Z+H_Z\vec{u}_d+H_ZH_X^T\vec{v}_X\,,\\
\vec{v}_Z'&=\vec{v}_Z\,.
\end{align}

In order to have a successful $Z$-stabilizer measurement we demand that $\vec{u}_Z'$ needs to lie in the set of correctable errors. The variance of $(\vec{u}_Z')_j$ is given by 
$\left(1+\|(H_Z)_{j,*}\|_2+\|(H_ZH_X^T)_{j,*}\|_2\right)\sigma^2_{sq}$.
 Also note that $H_Z H_X^T=0$ needs to hold in order to avoid error propagation between the GKP ancillas. However, up to now we only required that we are given a valid CSS code, which means that all stabilizer generators need to commute demanding $ H_ZH_X^T \mod D=0 $. These two conditions are equivalent to requiring that the symplectic form of any two rows of $H$ vanishes (without or with $\mod D$).
Therefore, it is useful to generalize the check matrix $H\in \mathbb{Z}_D^{ (n-k)\times 2n}$ to $\tilde{H}\in \mathbb{Z}^{(n-k)\times 2n}$, where $H\sim \tilde{H} \mod D$, $\sim$ denotes row equivalence with respect to the finite field $\mathbb{Z}_D$,  and furthermore, we need that 
 the symplectic form vanishes for any two distinct rows of $\tilde{H}$. 
 
 In a recent work \cite[Theorem 12]{invariant_stabilizer}  in the context of generalizing qubit to qudit codes, it was shown that it is always possible to find such a $\tilde{H}$. Thus, there is no error propagation anymore. However, this construction does not guarantee that the stabilizer weights remain small such that the noise actually coming from the data qudits may be amplified in the syndrome measurement.
 
 As one possible approach to reduce the stabilizer weights we can simply add rows of the matrix $\tilde{H}$ and try to minimize the stabilizer weights, which means we simply look for a different set of stabilizer generators. However, note that this approach is not feasible, because the problem is equivalent to being given a basis of a lattice and trying to find a different basis with minimal length and this is also known as the shortest basis problem on a lattice and which was shown to be NP-hard \cite{ShortestBasisProblem}.
 
 A different approach relies on fixing the stabilizer weight and trying to fulfill the symplectic condition. Here we will look at the cases $D=2$ and $D>2$ separately, because in the $D=2$ case $X$ and $Z$ are self-inverse giving us much more freedom while having the same stabilizer weight.
 
For $D>2$ it is not possible to sustain the minimal stabilizer weight from the canonical scheme and avoiding error propagation for arbitrary CSS codes, as it can be seen for the example of the $[D,D-2,2]_D$ error-detecting code with stabilizers $\prod_{j=1}^{D}X_j$ and  $\prod_{j=1}^{D}Z_j$. In order to sustain the minimal stabilizer weight, we can not modify the stabilizer generators, but their corresponding symplectic form does not vanish (without $\mod D$). However, for $D=2$ we can consider the stabilizers $X_1X_2$ and $Z_1Z_2^{-1}$ which still have minimal stabilizer weight, but their corresponding symplectic form vanishes (without $\mod D$).

For $D=2$ we can ideally fulfill the two conditions $\tilde{H}_Z\tilde{H}_X^T=0$ and $\|(\tilde{H}_{X/Z})_{j,*}\|_2= \|(H_{X/Z})_{j,*}|_2$ simultaneously. Let us now show some examples where we are able to fulfill both conditions.

As the first example let us consider the quantum parity code \cite{qpc}; this is a CSS code and the $Z$-stabilizers consist of weight 2 checks. Thus we choose $\tilde{H}_X=H_X$ and for $\tilde{H}_Z$ we use $H_Z$, but in each row we replace one of the two 1s by -1, thus the symplectic form is given by $1\times1+1\times(-1)=0$ (when it does not vanish trivially).
Also note that it is possible to define the quantum parity code for qudits.

Let us now consider 2-dimensional surface codes on lattices without boundary. If all face stabilizers have an even number of qubits in their support or if all vertex stabilizers have an even number of qubits in their support it is possible to achieve the optimal minimum. In order to do so  we will modify $\tilde{H}_{X/Z}$ for the type of stabilizers with even support (if it works for both faces and vertices we can choose) and do not change the other. Notice that face and vertex operators have either 0 or 2 common qubits in their support. As an example let us consider that our faces have even support. Instead of assigning each edge (corresponding qubit) the value 1 we assign $\pm1$ in an alternating way ('neighboring edges have different values'). Thus similar to the quantum parity code the symplectic form vanishes. Notice that this already includes many surface codes such as those with square, triangular, hexagonal tilings or even [4,5] tilings in hyperbolic geometry\cite{hyperbolic_surface_code}. 

However, also note that many surface and color codes have already been generalized from qubits to qudits by considering inverse Pauli operations\cite{qudit_surface,bombin_qudit_surface,primesquare_qudit_color,multidim_qudit_color}, implying that we can use their orientations to avoid error propagation and also obtain the optimal minimum.
\section{Comparison of different syndrome measurements}\label{sec:seven}
We have discussed two different approaches for obtaining the GKP syndrome information, namely an improvement of the Knill-Glancy scheme and an adaption of the error correction by teleportation scheme. Both schemes have the advantage of using no inline squeezing in contrast to schemes which make use of CSUM-gates, which are only implemented approximately. In general, the GKP Bell states needed for the teleportation scheme can be considered more expensive than the ancilla states for the Glancy-Knill scheme, because the former consist of a $2n$-mode entangled GKP state instead of two $n$-mode entangled states. However, for the case of even qudit dimension $D$ we have shown that it is possible to generate such a state by sending two $n$-mode entangled GKP states transversally through $n$ beam splitters. Because there are only  beam splitters and also no offline-squeezing we even get less noise than in the Knill-Glancy scheme.

For obtaining the high-level syndrome information we have considered three different schemes. Two of them (variations of the teleportation and the Knill-Glancy scheme) need no inline squeezing, but complicated ancilla states consisting of high-level encoded Bell states or (pre-squeezed) high-level Pauli eigenstates. These two schemes also have the advantage that we also obtain the GKP syndrome such that we only need to perform $2n$ measurements in order to obtain the full syndrome information. One might say that generating a high-level Bell state of a CSS code is not much more problematic than producing high-level Pauli-eigenstates because one could implement the logical CNOT via transversal CSUM gates, but there we also have the issue that we correlate or rather amplify the noise of different modes if we ignore the correlations.  However, in the third scheme (only for square GKP codes) we first use $2n$ measurements in order to correct displacements on the GKP qudits and then we perform the high-level stabilizer measurements by coupling ancilla states with the data-qudits via CSUM gates. This scheme has the advantage that the needed ancilla states are rather easy to generate, but one has various disadvantages: one needs inline squeezing operations, one has to use already $2n$ measurements in order to correct the small displacements and then additionally one has to measure the high-level stabilizers which also increases the noise of the already corrected data-qudits due to back propagation of errors  originating from the finitely squeezed ancillas.  

\section{Conclusion}\label{sec:conclusion}

In this article we have considered syndrome measurements of general GKP codes encoding qudits of dimension $D$ and their concatenation with stabilizer codes. We showed that we can obtain the full syndrome information of such an arbitrary $n$-mode code by making use of only $2n$ measurements. Furthermore, we discussed two schemes which allow us to obtain the GKP syndrome information by using either two suitable $n$-mode ancilla states or a single $2n$-mode GKP Bell state ancilla, transversal beam splitters and homodyne measurements.
 For the case of even qudit dimension $D$ we were able to show how GKP Bell states can be generated with transversal beam splitters and $n$-mode grid states. 
 
 Concerning the high-level syndrome information, we also propose two similar schemes without inline squeezing which give us the whole syndrome information with $2n$ homodyne measurements employing an ancilla state. We believe that not only for the Knill and Steane schemes as explicitly presented in this work, but for all fault-tolerant error correction schemes where the data modes are coupled by transversal CNOTs with an ancilla state (e.g. Shor states \cite[Sec. 4]{shorstate}) in order to perform the syndrome measurements of the higher code, one can additionally obtain the GKP syndrome information of all involved GKP codes.
  Moreover, we discussed error propagation in usual stabilizer measurements and also show that linear optical transformations leave the code distance of GKP codes and more generally error-correcting properties of codes against isotropic displacement noise invariant.
 We further analyzed the possibility of generating high level codewords by rectangular single-mode grid states and linear optics.
Besides this, we proposed an approach to calculate the logical error rates of a concatenation of a GKP code with a stabilizer code making use of the analog syndrome information where we calculate integrals instead of performing Monte-Carlo simulations.
Our main results can be summarized as follows:

\begin{itemize}
	\item GKP higher code syndrome detection: we proposed a minimal stabilizer set to be measured to obtain the full syndrome information,
	\item for logical qubits as well as qudits with non-prime dimensions the minimal measurement set is directly obtainable through Knill's error correction by teleportation on the higher level using higher GKP Bell states;	this directly provides an operational interpretation leading to	a possible implementation with transversal GKP qubit teleportations using beam splitters,
	\item for general logical qudits the minimal set can be derived via lattice theory,
	\item in a 2nd scheme, different from Knill's, we achieved the same	for higher code syndrome detections, generalizing known results for only the lower GKP level, still avoiding inline squeezing,
	\item  GKP higher code state generation: given higher $n$-mode GKP codes ($k<n$ qudits), we showed that the corresponding higher GKP Bell states cannot be obtained via transversal beam splitters; for arbitrary passive linear optics, it remains open,
	\item GKP higher code state generation: given copies of arbitrary rectangular single-mode grid states, we have shown that the codewords of the higher GKP codes can generally not be obtained via passive linear optics,
	\item GKP qudit Bell state generation: generalizing a known result for GKP qubits, we showed that for even qudit dimension the Bell states can be created from a number of suitable input grid states via transversal beam splitters  (this result includes states with $k=n$ qudits encoded into $n$ modes); whether this is also possible for odd qudit dimensions remains open.
\end{itemize}

\section*{Acknowledgment}
  We thank Daniel Miller for useful discussions about the generation of the GKP Bell states. We thank the BMBF in Germany for support via
  Q.Link.X/QR.X and the BMBF/EU for support via QuantERA/ShoQC.

\section*{Note}

At the final preparation stage of this work, Ref. \cite{conrad2021gottesmankitaevpreskill} was posted. Similar to our treatment that work also addresses the issue of a minimal stabilizer basis in higher GKP codes.
While there is also some overlap in terms of the methods used,
overall the two works are complementary,
where our work has a particular focus on linear-optical
realizations of the error correction schemes.

 \appendix
 \section{Minimal set of stabilizer generators}
 \label{app:minimal set}
 \begin{theorem}
 	For any GKP code ($n$ modes, arbitrary qudit dimension $D$) concatenated with an arbitrary stabilizer code it is possible to obtain the full syndrome information with $2n$ measurements.
 \end{theorem}
 
 \begin{proof}
 	It is well known that the phase space representation of the stabilizers of a GKP code forms a lattice $\mathcal{L} \subset \mathbb{R}^{2n}$. Similarly, the phase space representation of the set of operators commuting with the stabilizers $\mathcal{L}^\bot \subset \mathbb{R}^{2n}$ also forms a lattice \cite[Sec. VI]{gkp}. We can show that the phase space representation $\Lambda$ of the stabilizers of a GKP code concatenated with a higher-level stabilizer code also forms a lattice. For this we have to show that $\Lambda$ is a discrete, linear subgroup of $\mathbb{R}^{2n}$ and we will use the relation $\mathcal{L}\subseteq\Lambda\subset\mathcal{L}^\bot$ (the last relation holds because all stabilizers have to commute). Since we can obtain $\Lambda$ by adding additional points to $\mathcal{L}$ in a linear way, it is easy to see, that $\Lambda$ forms a linear subset of $\mathbb{R}^{2n}$. Since $\Lambda$ is a subset of $\mathcal{L}^\bot$ which is discrete (since it is a lattice), meaning that there exists an $\epsilon>0$ such that there is always at most one lattice point in an $\epsilon$ neighborhood, it is clear that $\Lambda$ is also discrete and therefore also forms a lattice. Every lattice has a basis \cite[Theorem 8]{latticenotes} and therefore we only have to measure the $2n$ operators corresponding to the lattice basis elements.
 \end{proof}

 \section{Linear optical decomposition of Bell states}
 \label{app:bell_decomp}
 
 Here we show that it is possible for arbitrary GKP codes with even qudit dimension $D$ to generate Bell states by mixing two GKP-like states at $n$ beam splitters transversally. Let us choose a fixed arbitrary GKP code (encoding $k=n$ qudits in $n$ modes) and let us write the logical Pauli operators as $\overline{X}_j=\exp(i \hat{\overline{x}}_j)$ ($j\in\{1,\dots,n\}$) implicitly defining $\hat{\overline{x}}_j$ and we do the same for $\hat{\overline{z}}_j$ with $\overline{Z}_j = \exp(i\hat{\bar{z}}_j)$  .
 
 In the next step the first index will number the logical operators of a GKP code and the second one will number the two codes. We start with the product state stabilized by the $4n$ stabilizers ($j$ takes every value in $\{1\dots n\}$)
 \begin{align*}
 \{&\exp(i \frac{D}{\sqrt{2}} \hat{\overline{z}}_{j,1}),\exp(i \sqrt{2}\hat{\overline{x}}_{j,1}),\\ &\exp(i \sqrt{2}\hat{\overline{z}}_{j,2}),\exp(i \frac{D}{\sqrt{2}}\hat{\overline{x}}_{j,2})\}\,.
 \end{align*}
For the special case of $n=1$ and $D=2$ we have the four stabilizers of the product state of two GKP 'qunaught' states (each representing a one-dimensional GKP space and hence a state with equal lattice spacing along $x$ and $p$, $\sqrt{2\pi}$). 

 After applying a 50:50 beam splitter transversally upon every pair of code states 1 and 2 for every $j$, we obtain: 
 \begin{align*}
 \{&\exp(i \frac{D}{2}\left(\hat{\overline{z}}_{j,1}+\hat{\overline{z}}_{j,2}\right)),\exp(i \left(\hat{\overline{x}}_{j,1}+\hat{\overline{x}}_{j,2}\right))\,,\\
 &\exp(i\left(\hat{\overline{z}}_{j,1}-\hat{\overline{z}}_{j,2}\right)),\exp(i \frac{D}{2}\left(\hat{\overline{x}}_{j,1}-\hat{\overline{x}}_{j,2}\right))
 \}\,.
 \end{align*}
 After a suitable multiplication (strictly assuming even $D$ to make sure an integer number of multiplications) of the stabilizers as discussed in the main text, we get
 \begin{align*}
 \{\exp(i D \hat{\overline{z}}_{j,1}), \exp(i\left(\hat{\overline{x}}_{j,1}+\hat{\overline{x}}_{j,2}\right)),\\
 \exp(i\left(\hat{\overline{z}}_{j,1}-\hat{\overline{z}}_{j,2}\right)), \exp(i D \hat{\overline{x}}_{j,1})	
 \}\,,
 \end{align*}
 where it is obvious that this set stabilizes GKP Bell states as this set contains $\overline{X}_1\overline{X}_2$ and $\overline{Z}_1\overline{Z}^{-1}_1$ which are the stabilizers of a Bell state and furthermore we have two independent stabilizer generators from the original GKP code. For the cases with odd $D$ we do not know whether GKP Bell states can be built from two $n$-mode code states with linear-optics. 
 
 When we consider a code encoding $k<n$ qudits in $n$ modes, unfortunately it is impossible to generate logical Bell states by coupling two product states by simple transversal beam splitters.
 In this case, the code space is defined by $4n$ independent stabilizer generators and $4k$ of them are proportional to logical Pauli operators. For these stabilizer generators we already know how the input stabilizers should look like. Thus, we only need to know how the remaining input stabilizers should look like. In order to obtain these we first consider the desired stabilizer generators and transform them by the inverse beam splitters (our beam splitters are self-inverse). Also notice that these stabilizer generators are independent (linear independent in the symplectic representation) and thus we only need to consider a pair of equivalent stabilizers of both codes:
 \begin{align*}
 	\{\exp(i \hat{g}_1),\exp(i (\hat{g}_1+\hat{g}_2))\}\\\rightarrow\{\exp(\frac{i}{\sqrt{2}}(\hat{g}_1+\hat{g}_2)),\exp(i\sqrt{2}\hat{g}_1)\}\,.
 \end{align*}
 It is obvious then that it is impossible to multiply the first stabilizer with the second one in such a way that the first stabilizer only acts on the modes belonging to code 2.

  \section{Knill error correction for qudits}\label{app:knill}
 Here we generalize the error correction by teleportation scheme proposed by Knill \cite{Knill_teleportation} from qubits to qudits. Although this scheme works for arbitrary qubit stabilizer codes, we have to restrict ourselves to CSS codes for the generalization to qudits, because the Pauli operators are not self-inverse anymore. 
 
 The projection operator onto the codespace with syndrome $s$ is given by 
 \begin{equation}
 \Pi(Q,e)=\prod_l\left(\sum_{j=0}^{D-1}(\exp(i\omega e_l )\hat{g}_l)^j\right)\,,
 \end{equation}
 where $\hat{g}_l$ is the $l$th stabilizer generator of the code represented by the matrix Q. 
 \begin{align}
 \Pi_2(Q,0) \ket{\Phi^+}^{\otimes n}_{12}\\
 =\Pi_2(Q,0) \Pi_2(Q,0)\ket{\Phi^+}^{\otimes n}_{12}\\
 = \Pi_2(Q,0) \Pi_1(\tilde{Q},0)\ket{\Phi^+}^{\otimes n}_{12}
 \end{align}
 
 In the first step, we wrote down the state which is needed to follow Knill's proof. We then try to simplify this expression. In the second line we used the idempotence of projection operators. In the next step we used that qudit Bell states are stabilized by the $X_1X_2$ and $Z_1Z_2^{-1}$. Therefore the projection onto the code represented by the matrix $Q$ with syndrome 0 on the second $n$ qudits is equivalent to a projection onto the code represented by $\tilde{Q}$  with syndrome 0 on the first $n$ qudits. Here $\tilde{Q}$ is given via $Q$ where all entries corresponding to $X$-operators are multiplied by -1. If $Q$ is a CSS code then this means that some rows have to be multiplied by -1 and their syndrome should yield 0. One can multiply these rows again by -1 to obtain $Q$, but the syndrome does not change. This can also be understood in the following way: all $X$-type operators in the stabilizer generators have been inverted. Thus for CSS codes the stabilizer group remains invariant. However, if $Q$ does not represent a CSS code it may describe a different code from $\tilde{Q}$. 
We checked it for the five-qudit (with stabilizers generators $X\otimes Z\otimes Z^{-1}\otimes X^{-1}\otimes \mathds{1}$ and cyclic permutations thereof) code that the stabilizer group generated by $Q$ does not equal the group generated by $\tilde{Q}$ for $D>2$ in general.
 
  The remaining proof is completely analogous to Knill's proof where he changes the order of the conditional Pauli operations and the projection operator, resulting in a changed syndrome and using the fact that the quantum teleportation protocol implements the identity.
 
 \section{Linear optical Knill-Glancy scheme for general GKP codes}
 \label{app:knill_glancy_gkp} 
 Let us consider an $n$ mode GKP code which encodes qudits of dimension $D$, but now without concatenation with a stabilizer code. Let us consider normalized quadrature operators $\hat{u}_j$ ($j\in \{1,\dots,n\}$) generating $X_j$ and normalized quadrature operators $\hat{v}_j$ generating $Z_j$. Thus we know that only $[\hat{u}_k,\hat{v}_k]\neq0$ and all other commutators vanish. Furthermore for a quadrature operator $\hat{s}$ there exists a symplectic representation as a $2n$-dimensional vector. We will refer to this symplectic representation as well as a measurement result of $\hat{s}$ as s, but it should always be clear from the context what the meaning is in each case. The quantity $\omega(\cdot,\cdot)$ denotes the canonical symplectic form.\\
 The stabilizers are then given by $X_j^D$ and $Z_j^D$ ($j\in\{1,\dots,n\}$) with
 \begin{align}
 	X_j=\exp\left(i \hat{u}_j\frac{1}{\sqrt{D\omega(u_j,v_j)}}\right)\,,\\
 		Z_j=\exp\left(i \hat{v}_j\frac{1}{ \sqrt{D\omega(u_j,v_j)}}\right)\,.
 \end{align}
 Without loss of generality we have assumed that $\omega(u_j,v_j)>0$ (the square GKP code is obtained with $\hat{u}_j=-\hat{p}_j$ and $\hat{u}_j=\hat{q}_j$)and $c_j\in \mathbb{R}_+$.
 In order to consider shift errors in the stabilizer formalism we use the identity
 \begin{align}
 	e^{i\hat{a}}e^{i\hat{b}}e^{-i\hat{a}}=e^{i \hat{b}}e^{-i 2\pi \omega(a,b)}\,.
 \end{align}
 
  Let us now briefly discuss how the stabilizers of a GKP code transform under shift errors $e^{i\hat{a}}$:
 \begin{align}
 \ket{\psi}=e^{i\hat{b}}\ket{\psi}\,,\\
 \ket{\tilde{\psi}}:=e^{i\hat{a}}\ket{\psi}=	e^{i\hat{a}}e^{i\hat{b}}\ket{\psi}=	e^{i\hat{a}}e^{i\hat{b}}e^{-i\hat{a}}e^{i\hat{a}}\ket{\psi}\\
 =e^{i(\hat{b}-2\pi\omega(a,b))}\ket{\tilde{\psi}}\,.
 \end{align}
 
 We will now show that we can apply the linear optical Knill-Glancy scheme to general GKP codes.
 In the first stage the $n$ data modes and the first $n$ ancilla modes are given by the following stabilizers assuming displacement errors with symplectic representation $e_1$ and $e_2$, subscripts 1 and 2 refer to the data and half of the ancilla modes, respectively,

 \begin{align*}
 	\exp\left(i\left(\hat{u}_{j,1}-2\pi\omega(e_1,u_j)\right)\sqrt{\frac{ D}{\omega(u_j,v_j)}}\right)\,,\\
 	\exp\left(i\left(\hat{v}_{j,1}-2\pi\omega(e_1,v_j)\right)\sqrt{\frac{D}{\omega(u_j,v_j)}}\right)\,,\\
 	\exp\left(i\left(\hat{u}_{j,2}-2\pi\omega(e_2,u_j)\right)\sqrt{\frac{1}{D \omega(u_j,v_j)}}\right)\,,\\
 	\exp\left(i\left(\hat{v}_{j,2}-2\pi\omega(e_2,v_j)\right)\sqrt{\frac{D}{\omega(u_j,v_j)}}\right)\,.
 \end{align*} 
 After applying the 50:50 beam splitters we obtain the following stabilizers:
 \begin{align*}
 \exp\left(i\left(\frac{\hat{\tilde{u}}_{j,1}+\hat{\tilde{u}}_{j,2}}{\sqrt{2}}-2\pi\omega(e_1,u_j)\right)\sqrt{\frac{D}{\omega(u_j,v_j)}}\right)\,,\\
 \exp\left(i\left(\frac{\hat{\tilde{v}}_{j,1}+\hat{\tilde{v}}_{j,2}}{\sqrt{2}}-2\pi\omega(e_1,v_j)\right)\sqrt{\frac{D}{\omega(u_j,v_j)}}\right)\,,\\
 \exp\left(i\left(\frac{\hat{\tilde{u}}_{j,1}-\hat{\tilde{u}}_{j,2}}{\sqrt{2}}-2\pi\omega(e_2,u_j)\right)\sqrt{\frac{1}{D \omega(u_j,v_j)}}\right)\,,\\
 \exp\left(i\left(\frac{\hat{\tilde{v}}_{j,1}-\hat{\tilde{v}}_{j,2}}{\sqrt{2}}-2\pi\omega(e_2,v_j)\right)\sqrt{\frac{D}{\omega(u_j,v_j)}}\right)\,.
 \end{align*}
 In the next step we perform measurements of $\hat{\tilde{v}}_{j,2}$ and the measurement outcomes $\tilde{v}_{j,2}$ give us partial information about $\omega(e_1-e_2,v_j)$ as it can be seen by the stabilizers (before the measurement)
 \begin{align*}
 	\exp\left(i\left(\sqrt{2}\hat{\tilde{v}}_{j,2}-2\pi\omega(e_1-e_2,v_j)\right)\sqrt{\frac{D}{\omega(u_j,v_j)}}\right)\,.
 \end{align*}
 After the measurement the stabilizers of the data qudits are given by 
 \begin{align*}
 \exp\left(i\left(\sqrt{2}\hat{\tilde{u}}_{j_1}-2\pi\omega(e_1+e_2,u_j)\right)\sqrt{\frac{D}{\omega(u_j,v_j)}}\right)\,,	\\
 \exp\left(i\left(\frac{\hat{\tilde{v}}_{j,1}+\tilde{v}_{j,2}}{\sqrt{2}}-2\pi\omega(e_1,v_j)\right)\sqrt{\frac{D}{\omega(u_j,v_j)}}\right)\,.
 \end{align*}
 We then apply a shift $\exp\left(i \hat{\tilde{u}}_{j,1}\frac{\tilde{v}_{j,2}}{2\pi\omega(u_j,v_j)}\right)$. The stabilizers in the second phase of the scheme are
 \begin{align*}
 	\exp\left(i\left(\sqrt{2}\hat{\tilde{u}}_{j,1}-2\pi\omega(e_1+e_2,u_j)\right)\sqrt{\frac{D}{\omega(u_j,v_j)}}\right)\,,\\
 	 	\exp\left(i\left(\frac{\hat{\tilde{v}}_{j,1}}{\sqrt{2}}-2\pi\omega(e_1,v_j)\right)\sqrt{\frac{D}{\omega(u_j,v_j)}}\right)\,,\\
 	 	\exp\left(i\left(\sqrt{2}\hat{\tilde{u}}_{j,3}-\sqrt{2}2\pi\omega(e_3,u_j)\right)\sqrt{\frac{D}{\omega(u_j,v_j)}}\right)\,,\\
 	 	\exp\left(i\left(\frac{\hat{\tilde{v}}_{j,1}}{\sqrt{2}}-\frac{2\pi\omega(e_3,v_j)}{\sqrt{2}}\right)\sqrt{\frac{1}{D\omega(u_j,v_j)}}\right) \,.	
 \end{align*}
 After applying the beam splitter we obtain 
 \begin{align*}
 \exp\left(i\left(\hat{\tilde{\tilde{u}}}_{j,1}+\hat{\tilde{\tilde{u}}}_{j,3}-2\pi\omega(e_1+e_2,u_j)\right)\sqrt{\frac{D}{\omega(u_j,v_j)}}\right)\,,\\
 \exp\left(i\left(\frac{\hat{\tilde{\tilde{v}}}_{j,1}+\hat{\tilde{\tilde{v}}}_{j,3}}{2}-2\pi\omega(e_1,v_j)\right)\sqrt{\frac{D}{\omega(u_j,v_j)}}\right)\,,\\
 \exp\left(i\left(\hat{\tilde{\tilde{u}}}_{j,1}-\hat{\tilde{\tilde{u}}}_{j,3}-\sqrt{2}2\pi\omega(e_3,u_j)\right)\sqrt{\frac{D}{\omega(u_j,v_j)}}\right)\,,\\
 \exp\left(i\left(\frac{\hat{\tilde{\tilde{v}}}_{j,1}-\hat{\tilde{\tilde{v}}}_{j,3}}{2}-\frac{2\pi\omega(e_3,v_j)}{\sqrt{2}}\right)\sqrt{\frac{1}{D\omega(u_j,v_j)}}\right)\,. 	
 \end{align*}
 We then measure the operators $\hat{\tilde{\tilde{u}}}_{j,3}$ which is again constrained by a stabilizer and this gives us partial information about $\omega(e_1+e_2-\sqrt{2}e_3,u_j)$.
 Thus, after the measurement the GKP code is stabilized by
 \begin{align*}
 	\exp\left(i\left(\hat{\tilde{\tilde{v}}}_{j,1}-2\pi\omega(e_1+\frac{e_3}{\sqrt{2}},v_j)\right)\sqrt{\frac{D}{\omega(u_j,v_j)}}\right)\,,\\
 	 	\exp\left(i\left(\hat{\tilde{\tilde{u}}}_{j,1}+\tilde{\tilde{u}}_{j,3}-2\pi\omega(e_1+e_2,u_j)\right)\sqrt{\frac{D}{\omega(u_j,v_j)}}\right)\,.
 \end{align*}
 Similarly as before we apply a shift $\exp\left(-i\hat{\tilde{\tilde{v}}}_{j,1}\frac{\tilde{\tilde{u}}_{j,3}}{2\pi\omega(u_j,v_j)}\right)$ in order to obtain the stabilizer 
 \begin{align*}
 	\exp\left(i\left(\hat{\tilde{\tilde{u}}}_{j,1}-2\pi\omega(e_1+e_2,u_j)\right)\sqrt{\frac{D}{\omega(u_j,v_j)}}\right)\,.
 \end{align*}
 A similar calculation can be done for the logical operators $\overline{X}$ and $\overline{Z}$. When doing this for $\overline{X}$ one can see that the logical operator transforms as
 \begin{align}
 	\exp\left(i\left(\hat{u}_{j,1}-2\pi\omega(e_1,u_j)\right) \sqrt{\frac{1}{D \omega(u_j,v_j)}}\right)\nonumber\\
 	\rightarrow\exp\left(i\left(\hat{\tilde{\tilde{u}}}_{j,1}-2\pi\omega(e_1+e_2,u_j)\right) \sqrt{\frac{1}{D \omega(u_j,v_j)}}\right)
 \end{align} which means we need to know $\omega(e_1+e_2,u_j)$ in order to perform the error correction, but we only know $\omega(e_1+e_2-\sqrt{2}e_3,u_j)\mod 2\pi\sqrt{\frac{\omega(u_j,v_j)}{D}}$ from our measurement results.

 Up to small displacements originating from the noise on the ancilla states, we now have the same state as before the error correction, but we can use our measurement results for a maximum likelihood estimation (which might also consider correlations between the measurement results) of $\omega(e_1,v_j)$ and $\omega(e_1+e_2,u_j)$ and apply correction shifts accordingly. Since we never use the periodicity of the exponential it is straightforward to see that a similar calculation also holds if one assumes that the data qudits are stabilized by either $X_j$ or $Z_j$. Thus logical errors can only occur if the maximum likelihood estimation fails. 
  \section{Linear optical Knill-Glancy scheme for concatenated CSS codes}  \label{app:knill_glancy_concat}
  
  Here we show that it is possible to obtain the full syndrome information in a scheme similar to the one described in the previous section. We only have to consider (squeezed) logical Pauli eigenstates of the high-level code instead of the GKP code.
  Since we consider a concatenation of a GKP code and a high-level code, we also have the GKP code stabilizers and additional ones from the high-level code. Thus, we obtain the syndrome information of the GKP code completely analogous as in the proof in the previous section and we only need to prove that we are able to obtain the syndrome information of the high-level code. However, notice that our new stabilizer set does not contain GKP Pauli operators, which were needed in order to ensure that the information encoded in the GKP code is not corrupted. This looks like a big problem, but actually we do not care whether the information in single GKP codes is corrupted. We only want that the information encoded in the concatenation of the GKP and the high-level code remains unchanged.
  This is achieved by having (squeezed) logical Pauli operators of the high-level code instead of those for the low-level GKP codes in the stabilizer group. 
   
   Let us now prove that we are able to obtain the syndrome information of the high-level code. The stablizers corresponding to the high-level code are given by (subscript $l$ numbers independent stabilizer generators of the high-level qudit code)
   \begin{align*}
   	\exp\left(i\sum_{j=1}^n\left(\hat{u}_{j,1}-2\pi\omega(e_1,u_j)\right)H_{jl}^{\hat{u}}\sqrt{\frac{1}{D\omega(u_j,v_j)}}\right)\,,\\
   	\exp\left(i\sum_{j=1}^n\left(\hat{v}_{j,1}-2\pi\omega(e_1,u_j)\right)H_{jl}^{\hat{v}}\sqrt{\frac{1}{D\omega(u_j,v_j)}}\right)\,,\\
   	\exp\left(i\sum_{j=1}^n\left(\hat{u}_{j,2}-2\pi\omega(e_2,u_j)\right)H_{jl}^{\hat{u}}\sqrt{\frac{1}{D\omega(u_j,v_j)}}\right)\,,\\
   	\exp\left(i\sum_{j=1}^n\left(\hat{v}_{j,2}-2\pi\omega(e_2,u_j)\right)H_{jl}^{\hat{v}}\sqrt{\frac{1}{D\omega(u_j,v_j)}}\right)\,.
   \end{align*}
   
   After applying the 50:50 beam splitter we obtain
   \begin{widetext}
   	   \begin{align*}
   		\exp\left(i\sum_{j=1}^n\left(\frac{\hat{\tilde{u}}_{j,1}+\hat{\tilde{u}}_{j,2}}{\sqrt{2}}-2\pi\omega(e_1,u_j)\right)H_{jl}^{\hat{u}}\sqrt{\frac{1}{D\omega(u_j,v_j)}}\right)\,,\\
   		\exp\left(i\sum_{j=1}^n\left(\frac{\hat{\tilde{v}}_{j,1}+\hat{\tilde{v}}_{j,2}}{\sqrt{2}}-2\pi\omega(e_1,u_j)\right)H_{jl}^{\hat{v}}\sqrt{\frac{1}{D\omega(u_j,v_j)}}\right)\,,\\
   		\exp\left(i\sum_{j=1}^n\left(\frac{\hat{\tilde{u}}_{j,1}-\hat{\tilde{u}}_{j,2}}{\sqrt{2}}-2\pi\omega(e_2,u_j)\right)H_{jl}^{\hat{u}}\sqrt{\frac{1}{D\omega(u_j,v_j)}}\right)\,,\\
   		\exp\left(i\sum_{j=1}^n\left(\frac{\hat{\tilde{u}}_{j,2}-\hat{\tilde{v}}_{j,2}}{\sqrt{2}}-2\pi\omega(e_2,u_j)\right)H_{jl}^{\hat{v}}\sqrt{\frac{1}{D\omega(u_j,v_j)}}\right)\,.
   	\end{align*}
   \end{widetext}

   We then measure $\hat{\tilde{v}}_{j,2}$ which is constrained by stabilizer conditions
   \begin{widetext}
	   \begin{align*}
		\exp\left(i\sum_{j=1}^n\left(\sqrt{2}\hat{\tilde{v}}_{j,2}-2\pi\omega(e_1-e_2,u_j)\right)H_{jl}^{\hat{v}}\sqrt{\frac{1}{D\omega(u_j,v_j)}}\right)\,,
	\end{align*}
   \end{widetext}

 giving us partial information about the displacement errors. If we perform an ideal formal stabilizer measurement we would learn the stabilizer
 \begin{align*}
 	\exp\left(i\sum_{j=1}^n\left(\hat{\tilde{v}}_{j,2}-2\pi\omega(e_1,u_j)\right)H_{jl}^{\hat{v}}\sqrt{\frac{1}{D\omega(u_j,v_j)}}\right)\,.
 \end{align*} Thus, up to a bit of noise originating from the noisy ancilla and a rescaling by a factor of $\sqrt{2}$, both approaches give the same information about the displacement errors. The state after the measurement, considering the new ancilla, is then given by
\begin{widetext}
\begin{align*}
	\exp\left(i\sum_{j=1}^n \left(\sqrt{2}\hat{\tilde{u}}_{j,1}-2\pi\omega(e_1+e_2,u_j)\right)H_{jl}^{\hat{u}}\sqrt{\frac{1}{D\omega(u_j,v_j)}}\right)\,,\\
	\exp\left(i\sum_{j=1}^n \left(\frac{\hat{\tilde{v}}_{j,1}+\tilde{v}_{j,2}}{\sqrt{2}}-2\pi\omega(e_1,v_j)\right)H_{jl}^{\hat{v}}\sqrt{\frac{1}{D\omega(u_j,v_j)}}\right)\,,\\
	\exp\left(i\sum_{j=1}^n \left(\sqrt{2}\hat{\tilde{u}}_{j,3}-\sqrt{2}2\pi\omega(e_3,u_j)\right)H_{jl}^{\hat{u}}\sqrt{\frac{1}{D\omega(u_j,v_j)}}\right)\,,\\
	\exp\left(i\sum_{j=1}^n \left(\frac{\hat{\tilde{v}}_{j,3}}{\sqrt{2}}-\frac{1}{\sqrt{2}}2\pi\omega(e_3,v_j)\right)H_{jl}^{\hat{v}}\sqrt{\frac{1}{D\omega(u_j,v_j)}}\right)\,.
\end{align*}
\end{widetext}
 
 By applying a corresponding displacement shift as in the previous section, we can remove the phase depending on $\tilde{v}_{j,2}$. After applying the second beam splitter we obtain
 \begin{widetext}
 \begin{align*}
 	\exp\left(i\sum_{j=1}^n \left(\hat{\tilde{\tilde{u}}}_{j,1}+\hat{\tilde{\tilde{u}}}_{j,3}-2\pi\omega(e_1+e_2,u_j)\right)H_{jl}^{\hat{u}}\sqrt{\frac{1}{D\omega(u_j,v_j)}}\right)\,,\\
 	\exp\left(i\sum_{j=1}^n \left(\frac{\hat{\tilde{\tilde{v}}}_{j,1}+\hat{\tilde{\tilde{v}}}_{j,3}}{2}-2\pi\omega(e_1,v_j)\right)H_{jl}^{\hat{v}}\sqrt{\frac{1}{D\omega(u_j,v_j)}}\right)\,,\\
 	\exp\left(i\sum_{j=1}^n \left(\hat{\tilde{\tilde{u}}}_{j,1}-\hat{\tilde{\tilde{u}}}_{j,3}-\sqrt{2}2\pi\omega(e_3,u_j)\right)H_{jl}^{\hat{u}}\sqrt{\frac{1}{D\omega(u_j,v_j)}}\right)\,,\\
 	\exp\left(i\sum_{j=1}^n \left(\frac{\hat{\tilde{\tilde{v}}}_{j,1}-\hat{\tilde{\tilde{v}}}_{j,3}}{2}-\frac{1}{\sqrt{2}}2\pi\omega(e_3,v_j)\right)H_{jl}^{\hat{v}}\sqrt{\frac{1}{D\omega(u_j,v_j)}}\right)\,.
 \end{align*}
 \end{widetext}
 
 We then measure $\hat{\tilde{\tilde{u}}}_{j,3}$ where we obtain partial information about the displacement errors due to the stabilizer constraint
 \begin{widetext}
  \begin{align*}
 	\exp\left(i\sum_{j=1}^n\left(2 \hat{\tilde{\tilde{u}}}_{j,2}-2\pi\omega(e_1+e_2-\sqrt{2}e_3,u_j)\right)H_{j,l}^{\hat{u}}\sqrt{\frac{1}{D \omega(u_j,v_j)}}\right)\,.
 \end{align*}
 \end{widetext}

 After the measurement the states are approximately (because of the small displacements on the ancillas) back to the code space and we have obtained the full syndrome information.
 The steps involving the logical Pauli operators showing that the logical information is not corrupted work completely analogous as in main text where we discuss the original Knill-Glancy scheme. 
 \section{Linear optics preserves code distance}
 \label{app:codedistance_preservation}
 (Passive) linear optical operations acting on $n$ modes are described by elements of the unitary group $U(n)$ acting on the mode operators. Using the 2-out-of-3 property \cite[p. 44]{symplectic_topology} of unitaries we see that $U(n)\cong O(2n)\cap Sp(2n)$. Therefore, the linear optical operation is represented by an orthogonal and symplectic matrix in the $2n$-dimensional phase space. 
 
 Let us consider a lattice $S\subset\mathbb{R}^{2n}$ where the symplectic form between any two lattice points yields an integer representing the commutation condition of stabilizer groups in the symplectic representation. This includes the case of general GKP codes and concatenations with higher level stabilizer codes. Furthermore, we define the dual (with respect to the symplectic form) lattice $\mathcal{L}^\bot(S)$ as the set of points whose symplectic form yields an integer with every point of the lattice $S$. The code distance of the corresponding code is then defined as \(\displaystyle\min_{\substack{u,v\in\mathcal{L}(S)/S\\
 u\neq v}}\norm{u-v}_2\) \footnote{We consider the 2-norm because for Gaussian noise the probability density only depends on the 2-norm of the symplectic error representation.}.

When we now apply a linear optical transformation to the corresponding state, we have to transform our lattice by multiplying it by an orthogonal and symplectic matrix $M$.
Therefore, the new lattice is given by $MS$, where the product is defined element-wise for every element of the group $S$. Since symplectic matrices do not change symplectic forms, it can be seen from the definition of the dual lattice that $M\mathcal{L}^\bot(S)\subseteq \mathcal{L}^\bot(MS)$. However, since $M$ is invertible, we even have equality between both sets (for a proof first apply $M$ and then $M^{-1}$ and obtain a sequence of subsets where the left and right sides are the same).
We now calculate the code distance after applying $M$ and see that it is left invariant since unitaries do not change the norm:
\begin{align*}
	\min_{\substack{u',v'\in\mathcal{L}(MS)/(MS)\\
			u'\neq v'}}\norm{u'-v'}_2=\min_{\substack{u,v\in\mathcal{L}(S)/S\\
			u\neq v}}\norm{M(u-v)}_2\\=\min_{\substack{u,v\in\mathcal{L}(S)/S\\
			u\neq v}}\norm{u-v}_2\,.
\end{align*}
Thus, linear optical transformations preserve the code distance of general GKP codes and we cannot hope to find a linear optical circuit transforming independent GKP codes into a high-level concatenated GKP code. 
However, it might still be possible that some codewords of the high-level code can be generated easily by individual GKP-like states and linear optics. One application of this possible loophole is the generation of the ancilla states that we need for our error correction schemes.

Furthermore, it is also easy to see that two general quantum error-correcting codes (not necessarily GKP codes) which are equivalent up to some linear-optical transformation have the same error-correcting properties against isotropic displacement error channels (e.g. i.i.d. Gaussian displacements). Instead of transforming the codes we can transform the noise channel accordingly. However, the isotropic displacement error channel is left invariant by the linear-optical transformation, because the probability distribution of the isotropic displacement noise channel only depends on the 2-norm of the displacement vector. This norm is preserved by the transformation as it acts as an orthogonal matrix in the phase space representation. As a consequence, the error-correcting properties of these two codes are the same against isotropic displacement error channels.  

\section{Exact calculation of analog information in the three-qubit repetition code}
\label{app:analoginfo}
When using a minimum-weight decoding scheme, we are applying a correction shift of minimum weight such that we recover the codespace, i.e. the combination of the error and correction shift is an element of the dual lattice $\mathcal{L}^\bot$. Since the three-qubit repetition code is a CSS code, we can correct position and momentum shifts independently, reducing the dimensionality of the computational problems by a factor of 2. Here we will also only discuss the position shifts as the momentum stabilizers are those of independent square lattice GKP qubits. The stabilizer generators and representatives of logical operators of the code are given by
\begin{align*}
	&\exp(i\sqrt{\pi}\left(\hat{q}_1-\hat{q}_2\right)) \,,\quad \exp(i\sqrt{\pi}\left(\hat{q}_2-\hat{q}_3\right)) \,,\quad \exp(i2\sqrt{\pi}\hat{q}_3)\,,\\
	&\exp(i2\sqrt{\pi}\hat{p}_1)\,,\quad\exp(i2\sqrt{\pi}\hat{p}_2)\,,\quad\exp(i2\sqrt{\pi}\hat{p}_3)\,,\\
	&\overline{X}=\exp(i\sqrt{\pi}\left(\hat{p}_1+\hat{p}_2+\hat{p}_3\right))\,,\\ &\overline{Z}=\quad\exp(i\sqrt{\pi}\hat{q}_1)\,.
\end{align*}

We can then decompose $\mathcal{L}^\bot=\mathcal{L}^\bot_\mathds{1}\cup \mathcal{L}^\bot_\mathds{X}$, corresponding to the represented operator $\overline{X}$.
In order to obtain the set of correctable errors we have to calculate the  voronoi cells, where each cell consists of all points being closest to a given lattice point, of $\mathcal{L}^\bot$ and consider the union of all  voronoi cells including a point in $\mathcal{L}^\bot_\mathds{1}$. This can easily be done by generating a finite-size lattice and using the scipy function \textit{scipy.spatial.Voronoi} for calculating the  voronoi cells. Using the translation invariance of the actual lattice we can then obtain all  voronoi cells by applying it to cells, which are not distorted due to finite-size effects. Since we consider a 3D lattice this can be visualized nicely and one sees that the correctable set of errors is given by a union of octagons where the elementary octagon is given by the convex span of the points $(\pm\frac{3\sqrt{\pi}}{2},0,0),(0,\pm\frac{3\sqrt{\pi}}{2},0),(0,0,\pm\frac{3\sqrt{\pi}}{2})$ and the other ones can be obtained by translations of $2\sqrt{\pi}(\mathbb{Z},\mathbb{Z},\mathbb{Z})$.

In order to obtain the probability of no bit-flip error we have to integrate the probability distribution of displacement errors over the set of correctable errors.
The overall set of correctable errors is too complicated for integration and therefore we integrate over a subset of octagons and obtain lower bounds on the probability of success (when considering the union we may not count some areas twice).

Let us now consider the most common case of i.i.d. Gaussian noise with a variance of $\sigma^2$. For the elementary octagon (and all others which are only displaced along one axis) we can split the octagon into two pyramids and consider new rotated integration variables, such that the base of the pyramid is aligned with the integration axes. This way we can do these integrations analytically and we are only left with the integral $\frac{2}{\sqrt{2\pi\sigma^2}}\int_0^{\frac{3\sqrt{\pi}}{2}} \exp(-\frac{z^2}{2\sigma^2})\erf\left(\left(\frac{3\sqrt{\pi}}{2}-z\right)\frac{1}{2\sigma}\right)^2 dz $ for the probability of no bit-flip error which then can be calculated numerically. The results of this calculation are shown in Fig. \ref{fig:gkprepetitioncodeanalog} and compared with the case where we do not make use of the analog GKP syndrome information.

\begin{figure}
	\centering
	\includegraphics[width=0.7\linewidth]{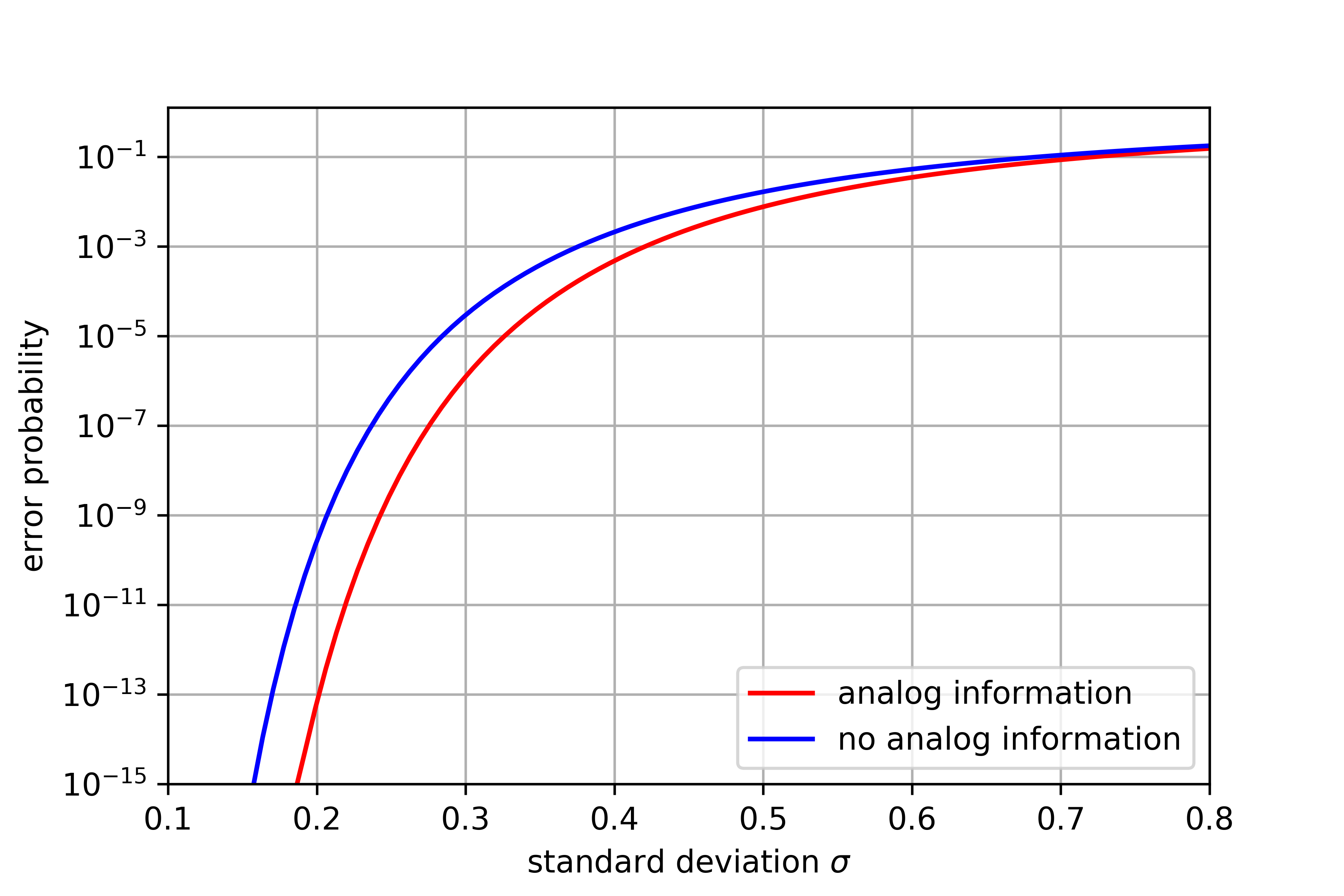}
	\caption{Logical bit-flip error rate of square GKP code concatenated with the three-qubit bit-flip code using our exact calculation. Our results are in good agreement with Ref. \cite[Fig. 2]{kosuke_minimal_measurement} where the results were obtained by a Monte-Carlo simulation. However, due to the simple numerical integration we are able to calculate small error rates where a Monte-Carlo approach would be infeasible.}
	\label{fig:gkprepetitioncodeanalog}
\end{figure}

\bibliography{ref}
\end{document}